\documentclass{article}

\usepackage{amsmath, amsfonts, amssymb, amsthm, bbm, url, mathtools}
\usepackage{braket}
\usepackage{microtype}

\usepackage{tikz,pgfplotstable}
\pgfplotsset{compat=1.18}

\usetikzlibrary{chains}
\usetikzlibrary{fit}
\usetikzlibrary{spy}

\usepackage{makecell}

\usepackage{epsfig}
\usetikzlibrary{shapes.symbols,patterns} 
\usepackage{pgfplots}

\usepackage{fullpage}
\usepackage{graphicx}
\usepackage{xcolor}

\usepackage[colorlinks=true, linkcolor=blue, citecolor=blue, urlcolor=blue]{hyperref}
\usepackage[capitalise]{cleveref}
\usepackage{cite}

\newcommand{\Z}{\mathbb{Z}}
\newcommand{\app}{\mathrm{app}}

\newcommand{\Q}{\mathbb{Q}}
\newcommand{\N}{\mathbb{N}}
\newcommand{\poly}{\mathrm{poly}}

\DeclareMathOperator{\tr}{\textup{tr}}

\DeclareMathOperator{\Wvec}{\operatorname{vec}}
\DeclareMathOperator{\dist}{\operatorname{dist}}
\DeclareMathOperator{\rank}{\operatorname{rank}}
\DeclareMathOperator{\bitlength}{bit-length}

\DeclareMathOperator{\R}{\mathbb{R}}
\DeclareMathOperator{\C}{\mathbb{C}}

\DeclareMathOperator{\A}{\mathcal{A}}
\DeclareMathOperator{\cH}{\mathcal{H}}
\DeclareMathOperator{\B}{\mathcal{B}}
\DeclareMathOperator{\D}{\mathcal{D}}
\DeclareMathOperator{\F}{\mathcal{F}}

\DeclareMathOperator{\cV}{\mathcal{V}}
\DeclareMathOperator{\W}{\mathcal{W}}

\newcommand{\footremember}[2]{%
    \footnote{#2}
    \newcounter{#1}
    \setcounter{#1}{\value{footnote}}
    \expandafter\gdef\csname frtext@#1\endcsname{#2}%
}
\makeatletter
\@ifundefined{pdfmatch}{%
  \newcommand{\footrecall}[1]{\footnote{\csname frtext@#1\endcsname}}%
}{%
  \newcommand{\footrecall}[1]{\footnotemark[\value{#1}]}%
}
\makeatother

\newcommand{\NP}{\texttt{NP}}
\renewcommand{\P}{\texttt{P}}
\newcommand{\QMA}{\texttt{QMA}}

\newcommand{\psucc}{\mathbb{P}_{\mathrm{succ}}}

\newcommand{\SEP}{\mathrm{SEP}}
\newcommand{\conv}{\mathrm{conv}}
\newcommand{\hsep}{\ell_{\SEP}}

\newcommand{\proj}[1]{|#1\rangle\!\langle#1|}
\newcommand{\ketbra}[2]{|#1\rangle\!\langle#2|}

\newtheorem{theorem}{Theorem}
\newtheorem{lemma}[theorem]{Lemma}
\newtheorem{proposition}[theorem]{Proposition}
\newtheorem{definition}[theorem]{Definition}
\newtheorem{corollary}[theorem]{Corollary}
\newtheorem{remark}[theorem]{Remark}

\title{Computational Aspects of the Trace Norm Contraction Coefficient}
\author{Idris Delsol\footremember{lyon}{Laboratoire de l'Informatique du Parallélisme, ENS Lyon, Inria, UCBL, F-69342 Lyon Cedex 07, France}
\and Omar Fawzi\footrecall{lyon}
\and Jan Kochanowski\footremember{paris}{Laboratoire Traitement et Communication de l'Information, Inria, Télécom Paris, Institut Polytechnique de Paris, 91120 Palaiseau, France}
\and Akshay Ramachandran\footremember{ubc}{Computer Science Department, University of British Columbia, Vancouver, BC, Canada, V6T 1Z4}}
\date{}

\begin{document}

\maketitle
{\renewcommand{\thefootnote}{}%
\footnotetext{Accepted for publication in \emph{IEEE Transactions on Information Theory}. DOI: 10.1109/TIT.2026.3743741. \copyright\ 2026 IEEE. Personal use of this material is permitted. Permission from IEEE must be obtained for all other uses, in any current or future media, including reprinting/republishing this material for advertising or promotional purposes, creating new collective works, for resale or redistribution to servers or lists, or reuse of any copyrighted component of this work in other works.}}

\begin{abstract}
    We show that approximating the trace norm contraction coefficient of a quantum channel within a constant factor is \NP-hard. Equivalently, this shows that determining the optimal success probability for encoding a bit in a quantum system undergoing noise is \NP-hard. This contrasts with the classical analogue of this problem that can clearly be solved efficiently. We also establish the \NP-hardness of deciding if the contraction coefficient is equal to $1$, i.e., the channel can perfectly preserve a bit.  As a consequence, deciding if a non-commutative graph has an independence number of at least $2$ is \NP-hard. In addition, we establish a converging hierarchy of semidefinite programming upper bounds on the contraction coefficient. 
\end{abstract}

\section{Introduction}\label{sec:intro}

A central question in information theory is how to optimally encode information for a given noise model. The objective of this paper is to study this question from an algorithmic point of view: given as input an explicit classical description of a quantum channel $\Phi$ that represents the noise and an integer $k$, what is the maximum success probability that can be achieved for sending $k$ messages using $\Phi$? When $\Phi$ is a classical channel, Barman and Fawzi~\cite{barman2017algorithmic} showed that this problem is \NP-hard to approximate within a factor greater than $1-e^{-1}$ and that efficient algorithms matching this hardness exist. When $\Phi$ is a quantum channel, much less is known about the complexity of the problem. When $\Phi$ is given implicitly as a quantum circuit on $n$ qubits, this problem is known to be \QMA-hard \cite{Beigi.2008,culf2023new}. Approximation algorithms have been proposed for some variants of the problem~\cite{Berta.2021,fawzi2019approximation,oufkir2024quantum} but the guarantees are in general significantly weaker than in the case of classical channels. In fact, the complexity of finding the optimal encoding of a single bit ($k=2$), which has a trivial efficient algorithm in the classical case, remains unknown.

\subsection{Main Results}

The main focus of this paper is to study the average success probability $\psucc(\Phi,2)$ for transmitting a uniformly distributed single-bit message over the quantum channel $\Phi$, when an optimal encoder and decoder are used.
It is simple to connect this expression to the trace norm contraction coefficient of the channel $\Phi$  (see Proposition~\ref{prop:optimal_succ_contraction} for details):
\begin{equation*}
    \psucc(\Phi,2) = \frac{1}{2} + \frac{1}{2} \eta_{\tr}(\Phi)
\end{equation*}
where the contraction coefficient is defined as follows.
\begin{definition}\label{def:contraction_coefficient}
The \emph{trace norm contraction coefficient} of a channel $\Phi$ is 
\begin{equation}\label{eq:contraction_coefficient}
    \eta_{\tr}(\Phi) = \max_{\substack{\rho,\sigma \in \D \\ \rho \not= \sigma}} \frac{\|\Phi(\rho) - \Phi(\sigma)\|_1}{\|\rho - \sigma\|_1},
\end{equation}
where $\D$ is the set of density operators on the input space of $\Phi$. 
\end{definition}

As such, computing the trace norm contraction coefficient is equivalent to computing the optimal success probability and we focus in the following on the contraction coefficient. We start by establishing hardness results. The first result establishes hardness of approximation within a constant factor. We say that the algorithm $\A$ approximates $f$ within a factor $0 < q \leq 1$ if for every input $x$ we have $q \cdot f(x) \leq \A(x) \leq f(x)$.
\begin{theorem}[General hardness]\label{thm:NP_contraction_coef}
Given a channel $\Phi$, its contraction coefficient $\eta_{\tr}(\Phi)$ is \NP-hard to approximate to within a factor $1/\sqrt{2} + \varepsilon$ for any $\varepsilon > 0$. The hardness result holds even if $\Phi$ is promised to be unital.

Furthermore, even in the restricted setting where $\Phi$ is promised to be a quantum-classical (i.e., quantum input and classical output) and unital channel, the contraction coefficient is \NP-hard to approximate to within a factor $\sqrt{2/\pi} + \varepsilon$ for any $\varepsilon > 0$.
\end{theorem}

This result is proven in Section~\ref{sec:NP_general_case} and is based on a reduction from the Little Grothendieck Problem. We remark that for some special classes of channels, the contraction coefficient can be approximated efficiently. For example, if the channel $\Phi$ is composed $t$ times with $t \to \infty$, the optimal success probability and hence the contraction coefficient can be efficiently computed~\cite{Fawzi.2025,Singh.2024}.

Our second result shows that the question of deciding whether $\eta_{\tr}(\Phi) = 1$ or $\eta_{\tr}(\Phi) < 1$ is also \NP-hard. 
\begin{theorem}[Hardness for the complete case]\label{thm:NP_complete_case}
Given a quantum channel $\Phi$, it is \NP-hard to distinguish between $\eta_{\tr}(\Phi) = 1$ and $\eta_{\tr}(\Phi) \leq 1 - \frac{1}{\poly(n)}$, where $n$ is the dimension of the input Hilbert space. 
\end{theorem}
In terms of channel coding, this corresponds to deciding the existence of a perfect, i.e., zero-error, encoding of a bit. The existence of a zero-error code of size $k$ for a quantum channel $\Phi$ is equivalent to the independence number of the corresponding non-commutative confusability graph of $\Phi$ being greater than or equal to $k$~\cite{Duan.2013}. A non-commutative graph is a subspace $S$ of matrices containing the identity and satisfying $S = S^*$, i.e., being self-adjoint, which is more commonly known as an operator system.\footnote{The way to see an undirected graph $G = (V,E)$ in this setting is by choosing $S = \operatorname{Span}\{ E_{i,j} \mid i, j \in V, i = j \text{ or } (i,j) \in E \}$, where $E_{i,j}$ is the matrix containing $1$ in position $(i,j)$ and zero everywhere else.} Theorem~\ref{thm:NP_complete_case} shows that deciding whether a non-commutative graph has independence number at least $2$ is \NP-hard. Note that in the model where the channel is described implicitly by a circuit, the hardness of this problem was studied in depth in~\cite{culf2023new}.

A consequence of this work is that known efficiently computable bounds on the contraction coefficient such as the quantum Doeblin coefficient~\cite{Hirche2024.2,Hirche2024.1} cannot always be tight, assuming $\P \neq \NP$. In fact, we give in~\cref{prop:counter_example_Hirche} an explicit example of a channel $\Phi$ where the Doeblin coefficient fails to capture $\eta_{\tr}(\Phi) < 1$.

In the last section of this article, we consider $\psucc(\Phi,k)$, the average success probability of retrieving a message drawn uniformly at random in $\{1,\dots,k\}$ and sent through $\Phi$ using optimal quantum encoding and decoding. In our setting, a quantum encoding consists of $k$ states $\rho_1, \dots, \rho_k$, while the decoder is defined by a POVM $\{M_1,\dots, M_k\}$, where the $M_i$ are positive semidefinite operators summing to the identity, i.e., $\sum_i M_i = I$. By Born's rule~\cite{Born.1926}, the probability of successfully decoding the message $i \in \{1, \dots, k\}$ sent through the channel $\Phi$ is $\tr(M_i \Phi(\rho_i))$, hence we mathematically define
\begin{equation}\label{eq:proba_success}
\psucc(\Phi,k) = \max_{\substack{ \{M_i \mid i \in [k]\} \text{ POVM},\\\{\rho_i \in\D(\cH) \mid i \in [k]\}}} \frac{1}{k}\sum_{i=1}^k \tr(M_i \Phi(\rho_i)).
\end{equation}
Note that, using this notation, Shannon's channel coding theorem~\cite{Shannon.1948} and its quantum counterpart, due to Holevo, Schumacher and Westmoreland~\cite{Holevo.1998,Schumacher.1997}, read
\begin{equation}
\limsup_{n\to +\infty} \psucc(\Phi^{\otimes n}, 2^{\lceil nc \rceil}) = 
\begin{cases}
1 \text{ if } c < C, \\
< 1 \text{ if } c > C,
\end{cases}
\end{equation}
where $C$ is the classical capacity of $\Phi$.

Our final result consists of applying the methodology of~\cite{Berta.2021} for constrained bilinear problems to obtain converging hierarchies of semidefinite programming upper bounds on the success probability $\psucc(\Phi,k)$, for each possible number of messages $k$.
This yields an upper bound on the contraction coefficient (Theorem~\ref{thm:sdp_POVM}) when we take $k=2$. We illustrate this hierarchy by applying it to some simple channels (\cref{fig:amplitude_damping}, \cref{fig:depolarizing}) in the case $k=2$. On all the examples we tried, the first level of our hierarchy gives bounds that are tighter than the quantum Doeblin coefficient introduced in \cite{Hirche2024.1}. We note that \cite[Corollary 3]{Hirche2024.2} considered a semidefinite programming hierarchy for a variant of the Doeblin coefficient called the induced Doeblin coefficient. This hierarchy is different from ours as it cannot converge to the contraction coefficient. In fact, as shown in~\cref{prop:counter_example_Hirche}, the induced Doeblin coefficient fails to capture $\eta_{\tr}(\Phi) < 1$.

\section{Notation and basic statements}\label{sec:notations}

In this paper, all Hilbert spaces $\cH$ are finite dimensional. We write $\B(\cH,\cH')$ for the set of linear operators from $\cH$ to $\cH'$ and $\B(\cH) := \B(\cH,\cH)$. The state of a quantum system with associated Hilbert space $\cH$ is described by a density operator, i.e., a positive and unit trace operator in $\B(\cH)$. 
We denote by $\D(\cH)$ the set of such density operators. 
Physically possible transformations are mathematically represented by completely positive trace-preserving maps, also called quantum channels. 
For every quantum channel $\Phi : \B(\cH) \to \B(\cH')$, we can find a set of operators $\{K_i \in \B(\cH,\cH')\mid i \in [l]\}$ satisfying $\sum_i K_i^* K_i = I$ such that for all $x \in \B(\cH)$:
\begin{equation}\label{eq:Kraus_op}
\begin{aligned}
\Phi&(x) = \sum_iK_ixK^*_i.
\end{aligned}
\end{equation}
These operators are called Kraus operators of the channel $\Phi$ and there is a Kraus representation of $\Phi$ using at most $\dim(\cH)\dim(\cH')$ such operators (see~\cite[Theorem~2.22]{watrous.2018}). The Choi state of $\Phi$ is denoted as
\begin{align}
\label{eq:choi-state}
J(\Phi) = ( \text{Id}_{\cH} \otimes \Phi)(\proj{\psi^+}),
\end{align}
where $\ket{\psi^+} = \frac{1}{\sqrt{\dim(\cH)}} \sum_{i=1}^{\dim(\cH)} \ket{i} \ket{i}$ for some fixed orthonormal basis $\ket{i}$ of $\cH$. Given $\Phi$ a channel, we denote by $A$ its input system and by $B$ its output system, so that $J(\Phi)$ is a bipartite state on the compound system $(A,B)$. We will therefore write $J(\Phi)_{AB}$ for the Choi state of $\Phi$, whenever the systems on which $J(\Phi)$ acts may be ambiguous, especially in \cref{sec:sdp_hierarchy}. Given $S$ a physical system, we use the shorthand notation $d_S := \dim(S)$. We write $\Phi^*$ for the adjoint (also called the dual) of the map $\Phi : \B(\cH) \to \B(\cH')$, i.e., we have $\langle \Phi(X), Y \rangle = \langle X, \Phi^*(Y) \rangle$ for all $X \in \B(\cH), Y \in \B(\cH')$. Here, $\langle .,.\rangle$ is the Hilbert--Schmidt inner product $\langle A, B \rangle = \tr(A^* B)$. Note that $\Phi$ is a quantum channel if and only if $\Phi^*$ is completely positive and unital, i.e., $\Phi^*(I) = I$.
Also, given $X, Y \in \B(\cH)$ Hermitian operators, we write $X \succeq Y$ whenever $X - Y$ is positive semidefinite.

It will be convenient to make use of the following correspondence between $\B(\cH',\cH)$ and $\cH \otimes \cH'$, where we denote $p = \dim(\cH)$ and $q=\dim(\cH')$. Let $\{\ket{i} \mid i\in [p]\}$ and $\{\ket{j} \mid j \in [q]\}$ be fixed orthonormal bases of $\cH$ and $\cH'$, then we define the linear map $\Wvec : \B(\cH',\cH) \to \cH \otimes \cH'$ as
\begin{equation}\label{eq:def_vec}
\Wvec(\ketbra{i}{j}) := \ket{i}\otimes \ket{j}.
\end{equation}
This map is an isometry in the sense that for all $x,y \in \B(\cH',\cH)$, $\langle x,y\rangle = \langle \Wvec(x), \Wvec(y) \rangle$.

The set of all quantum channels is uncountable and hence, it is not recursively enumerable.
Therefore, a rigorous statement of our Theorems~\ref{thm:NP_contraction_coef} and~\ref{thm:NP_complete_case} requires introducing a restricted set of quantum channels, whose elements can be represented as inputs to a Turing machine.
As stated above, a channel $\Phi : \B(\cH) \to \B(\cH')$ is unambiguously specified by a set of at most $dd'$ Kraus operators in $\B(\cH,\cH')$, where $d$ and $d'$ are the dimensions of $\cH$ and $\cH'$, respectively.
Writing these operators in the computational basis yields a representation of $\Phi$ by a list $L_\Phi$ of at most $(dd')^2$ complex numbers.
We consider a restricted set of channels $\Phi$ for which $L_\Phi$ contains only \emph{algebraic numbers}.
These numbers can be represented as input to a Turing machine, as we discuss now.

A complex number $z \in \C$ is an \emph{algebraic number} if it is the root of a non-zero polynomial $P \in \Z[X]$.
The set of algebraic numbers is denoted $\overline{\Q}$.
Examples of algebraic numbers include $i$, $\sqrt{q}$ for $q \in \Q$ and $q \geq 0$, or the eigenvalues of matrices with algebraic entries.
As discussed in~\cite[Section~4.2]{Cohen.1993}, there exist four main representations for algebraic numbers. 
In this article, we will use two of them, the first one being \emph{absolute} and the second \emph{relative}, in the sense that the second expresses $z \in \overline{\Q}$ as a rational linear combination of powers of a fixed $\theta \in \overline{\Q}$, whereas the first does not rely on the specification of any other algebraic number.  
Let us first introduce the absolute representation of algebraic numbers.
For all $z\in \overline{\Q}$, there exists a polynomial $P_z \in \Z[X]$ such that $z$ is a root of $P_z$ and for every $Q \in \Z[X]$ vanishing on $z$, $P_z$ divides $Q$ in $\Z[X]$. 
This polynomial is called the \emph{minimal polynomial} of $z$ and its degree is called the \emph{degree of} $z$.
Hence an algebraic number can be represented by the list of coefficients of its minimal polynomial and a complex rational $z_{\app} \in \Q + i\Q$ that is closer to $z$ than to any other root of $P_z$.

We then move on to the relative representation.
Given $\theta \in \overline{\Q}$ of degree $t$, the \emph{field extension} associated to $\theta$, denoted $\Q(\theta)$, is the set of all complex numbers $z$, such that there is $b \in \N^*$, $a_0,\dots,a_{t-1} \in \Z$ with $\gcd(b,(a_j)_{0\leq j < t}) = 1$ satisfying

\begin{equation}\label{eq:field_extension}
    z = \frac{1}{b}\sum_{j=0}^{t-1} a_j \theta^{j}.
\end{equation}
The elements of $\Q(\theta)$ are algebraic numbers.
Indeed, $\Q(\theta)$ is both a field and a $t$-dimensional $\Q$-linear space (see~\cite[Section~4]{Cohen.1993}). 
Hence, if $z \in \Q(\theta)$, the $t+1$ consecutive powers $1,z,\dots,z^t$ are in $\Q(\theta)$.
Since $\Q(\theta)$ has dimension $t$ over $\Q$, there must exist a non-trivial linear combination over $\Q$ of the  $(z^j)_{0 \leq j \leq t}$ which vanishes.
Then, $z$ is a root of the polynomial whose coefficients are those of this linear combination, so that $z$ is algebraic of degree at most $t$.
We can then represent $z \in \Q(\theta)$ by the sequence of the integers $(a_j)_{0 \leq j \leq t-1}$ and $b$ such that $z = \tfrac{1}{b}\sum_{j=0}^{t-1}a_j\theta^j$ with $\gcd(b,a_0,\dots,a_{t-1}) = 1$, as per \eqref{eq:field_extension}.
As it is convenient to keep track of the absolute representation of $\theta$ itself, so as to have an absolute reference point to uniquely specify $z$, the binary length of $z$ in its representation relative to $\theta$ is the sum of the binary length of the integers $(a_j)_{0 \leq j \leq t-1}$ and $b$ with that of $\theta$ in its absolute representation.
We can generalise this construction when a finite set of algebraic numbers $\theta_1,\dots,\theta_k \in \overline{\Q}$ of respective degrees $t_1,\dots,t_k$ is given.
The field extension generated by $\theta_1,\dots,\theta_k$ is then the set of complex numbers $z$ satisfying the following condition.
There exist $b \in \N^*$ and a set of integers $(a_{\mathbf{i}})_{\mathbf{i} \in I}$, with $I = \{\mathbf{i} = (i_1,\dots,i_k) \mid 0 \leq i_j < t_j, j\in [k]\}$, such that $\gcd(b,(a_{\mathbf{i}})_{\mathbf{i} \in I}) = 1$, and

\begin{equation}\label{eq:field_extension_multiple_generators}
    z = \frac{1}{b}\sum_{\mathbf{i} \in I} a_{\mathbf{i}}\prod_{j=1}^k \theta_j^{i_j}.
\end{equation}
The elements of $\Q(\theta_1,\dots,\theta_k)$ are algebraic numbers of degree at most $\prod_{j=1}^k t_j$.
Note that if the monomials $\prod_{j=1}^k \theta_j^{i_j}$, $\mathbf{i} \in I$, are not linearly independent over $\Q$, the number of coefficients required to uniquely specify $z \in \Q(\theta_1,\dots,\theta_k)$ may be strictly less than $\prod_{j=1}^k t_j$.
Note also that by the primitive element theorem~\cite[Theorem~4.1.8]{Cohen.1993}, there exists $\gamma \in \overline{\Q}$ such that $\Q(\theta_1,\dots,\theta_k) = \Q(\gamma)$.

We fix a representation of channels whose Kraus operators have algebraic entries in the following definition.
To this end, we first fix a representation of matrices with algebraic entries.

\begin{definition}\label{def:rep_matrix_number_field}
Let $\theta_1,\dots,\theta_k \in \overline{\Q}$ and $A$ be an $n \times m$ matrix with entries in $\Q(\theta_1,\dots,\theta_k)$, i.e., $A \in \Q(\theta_1,\dots,\theta_k)^{n \times m}$. 
We represent $A$ as the variable length nested list $(T_1,\dots,T_k, C_{1,1}, \dots,C_{n,m})$, where:
\begin{itemize}
    \item For $\alpha \in [k]$, $T_\alpha$ is a variable length list of integers containing the absolute representation of $\theta_{\alpha}$. We have $T_{\alpha} = (a_0,\dots,a_d,p_\app,q_\app,r_\app,s_\app)$, where $(a_j)_{0 \leq j \leq d}$ are the coefficients of the minimal polynomial of $\theta_\alpha$, and $\tfrac{p_\app}{q_\app} + i\tfrac{r_\app}{s_\app}$ is the approximant of $\theta_\alpha$.

    \item For $\beta \in [n]$, $\gamma \in [m]$, if $[A]_{\beta,\gamma}$ is non-zero, it is represented by the list $C_{\beta,\gamma} = (B_{\beta,\gamma}, D_{\beta,\gamma})$, where both $B_{\beta,\gamma}$ and $D_{\beta,\gamma}$ are lists of integers.
    $B_{\beta,\gamma}$ is a list of indices in $[k]$.
    $D_{\beta,\gamma}$ has length $1+\prod_{\alpha \in B_{\beta,\gamma}} \deg(\theta_\alpha)$.
    Its elements $((a_{\mathbf{i}})_{\mathbf{i} \in I}, b)$ satisfy $\gcd((a_{\mathbf{i}})_{\mathbf{i} \in I},b) = 1$ and $[A]_{\beta,\gamma} = \tfrac{1}{b}\sum_{\mathbf{i} \in I} a_{\mathbf{i}} \prod_{\alpha \in B_{\beta,\gamma}} \theta_\alpha^{i_\alpha}$, with $I = \{\mathbf{i} = (i_\alpha)_{\alpha \in B_{\beta,\gamma}} \mid 0 \leq i_\alpha \leq \deg(\theta_\alpha)-1, \alpha \in B_{\beta,\gamma}\}$.  
    If $[A]_{\beta ,\gamma}$ is zero, $C_{\beta,\gamma}$ is empty.
\end{itemize}
\end{definition}

\begin{definition}\label{def:representation_channel}
Let $\Phi : \B(\cH) \to \B(\cH')$ be a quantum channel with Kraus decomposition $\{K_j \in \B(\cH,\cH') \mid j \in [l]\}$, with $l \leq hh'$ and $(h, h')=(\dim(\cH),\dim(\cH'))$, such that, for all $j \in [l]$, $\alpha \in [h']$, $\beta \in [h]$, $[K_j]_{\alpha,\beta} \in \Q(\theta_1^{(j)},\dots,\theta_{k^{(j)}}^{(j)})$. 
Then, we represent $\Phi$ as the nested list of variable length $(L_1,\dots,L_l)$, where for each $j \in [l]$, $L_j$ represents $K_j$ as a matrix as in Definition~\ref{def:rep_matrix_number_field}.
\end{definition}

Note that the representation of $\Phi$ under this definition is not unique. 
This is because, first, the Kraus decomposition of $\Phi$ is given only up to a unitary mixing of the Kraus operators, and, second, we have the freedom to choose which algebraic numbers $\theta^{(j)}_1, \dots, \theta^{(j)}_{k^{(j)}}$ are represented using the absolute representation.
Also, we refer the reader to~\cite[Definition~1.4.1]{Winkler.1996} for the formal definition of a \emph{list} as a data structure.

We now analyse the space complexity of the representations we introduced in Definitions~\ref{def:rep_matrix_number_field} and~\ref{def:representation_channel}. 
Let $\Phi$ be represented by the nested lists $(L_1,\dots,L_l)$ as in Definition~\ref{def:representation_channel}.
For each $j \in [l]$, $L_j = (T^{(j)}_1,\dots,T^{(j)}_{k^{(j)}},C^{(j)}_{1,1},\dots,C^{(j)}_{h',h})$ by Definition~\ref{def:rep_matrix_number_field}, hence we compute the binary length of the $(T^{(j)}_\alpha)_{\alpha \in [k^{(j)}]}$ and $(C^{(j)}_{\beta,\gamma})_{\beta\in [h'],\gamma\in[h]}$ separately.
As they both involve listing integers, we recall that given $x \in \Z \backslash \{0\}$, $x$ can be represented on $\lfloor \log_2(|x|) \rfloor + 2$ bits accounting for its sign.
For $\alpha \in [k^{(j)}]$, $T^{(j)}_\alpha$ contains the coefficients of the minimal polynomial of $\theta^{(j)}_\alpha$, as well as $\tfrac{p_\app}{q_\app} + i\tfrac{r_\app}{s_\app}$ an approximant of $\theta^{(j)}_\alpha$.
For $P \in \Z[X]$, we write $\|P\|_{\ell_p}$ for the $\ell_p$-norm of the vector whose coefficients are those of $P$. 
Then, the list of the coefficients of $P$ has binary length $O(\deg(P)\log(\|P\|_{\ell_\infty}+1))$. 
We show in Appendix~\ref{sec_apdx:comput_model}, as Corollary~\ref{coro:mignotte}, that for every root $z$ of $P$, we can choose an approximant $z_\app$ of $z$ that specifies it among the other roots of $P$ using $O(\deg(P)\log(\deg(P)\,\|P\|_{\ell_2}+1))$ bits.

We write $P^{(j)}_\alpha\in \Z[X]$ for the minimal polynomial of $\theta^{(j)}_\alpha$ for $\alpha \in [k^{(j)}]$, so that the binary length of $T^{(j)}_\alpha$ is $O(\deg(P^{(j)}_\alpha)\log(\deg(P^{(j)}_\alpha)\|P^{(j)}_\alpha\|_{\ell_2}+1))$.

Let $d^{(j)} = \max_{\alpha \in [k^{(j)}]} \deg(P^{(j)}_\alpha)$ and $N^{(j)} = \max_{\alpha \in [k^{(j)}]} \|P^{(j)}_\alpha\|_{\ell_2}$.
Then the binary length of $(T^{(j)}_1,\dots,T^{(j)}_{k^{(j)}})$ is $O(k^{(j)}d^{(j)}\log(d^{(j)}N^{(j)}+1))$.
Now, let $\beta \in [h']$ and $\gamma \in [h]$.
$C^{(j)}_{\beta,\gamma}$ is either the empty list or the list $(B^{(j)}_{\beta,\gamma}, D^{(j)}_{\beta,\gamma})$.
$B^{(j)}_{\beta,\gamma}$ has length at most $k^{(j)}$ and its elements are integers in $[k^{(j)}]$.
Thus, the binary length of $B^{(j)}_{\beta,\gamma}$ is $O(k^{(j)}\log(|k^{(j)}|+1))$.
Let $\Delta^{(j)}_{\beta,\gamma} = \max_{x \in D^{(j)}_{\beta,\gamma}} |x|$ and $b^{(j)}_{\beta,\gamma} = |B^{(j)}_{\beta,\gamma}|$.
The list $D^{(j)}_{\beta,\gamma}$ contains $1+\prod_{\alpha \in B^{(j)}_{\beta,\gamma}} \deg(\theta_\alpha) \leq 1+(d^{(j)})^{b^{(j)}_{\beta,\gamma}}$ elements.
Hence, the binary length of $D^{(j)}_{\beta,\gamma}$ is $O((d^{(j)})^{b^{(j)}_{\beta,\gamma}}\log(\Delta^{(j)}_{\beta,\gamma}+1))$.
Let $\Delta^{(j)} = \max_{\beta,\gamma} \Delta^{(j)}_{\beta,\gamma}$ and $b^{(j)} = \max_{\beta,\gamma} b_{\beta,\gamma}^{(j)}$.
The binary length of the whole list $L_j$, denoted $\bitlength(L_j)$, then satisfies
\begin{equation}\label{eq:bit_length_matrix}
    \bitlength(L_j) = O\bigl(k^{(j)}d^{(j)}\log(d^{(j)}N^{(j)}+1) + hh'k^{(j)}\log(k^{(j)}+1) + hh'd^{(j)b^{(j)}}\log(\Delta^{(j)}+1)\bigr).
\end{equation}

Finally, let us write $k = \max_j k^{(j)}$, $d = \max_j d^{(j)}$, $N = \max_j N^{(j)}$, $\Delta = \max_j \Delta^{(j)}$, $b = \max_j b^{(j)}$, then the binary length of $\Phi$, denoted $\bitlength(\Phi)$, in the representation of Definition~\ref{def:representation_channel} satisfies

\begin{equation}\label{eq:bit_length_channel}
    \bitlength(\Phi) = O\Bigl(hh'\bigl[kd\log(dN+1) + hh'k\log(k+1) + hh' d^b\log(\Delta+1)\bigr]\Bigr).
\end{equation}

Note that the binary length of the representation of $\Phi$ is polynomial in the total number of generators (at most $hh'k$), but exponential in the maximum number of generators $b$ on which a single entry of $\Phi$ depends.

While we omit implementation details, the representation of Definitions~\ref{def:rep_matrix_number_field} and~\ref{def:representation_channel} enables an encoding of the spectral decomposition of Hermitian matrices with algebraic entries.

\begin{proposition}\label{prop:rep_spectral_decompo}
    Let $\theta_1,\dots,\theta_k \in \overline{\Q}$ of respective degrees $(d_j)_{j\in [k]}$, and $A \in \Q(\theta_1,\dots,\theta_k)^{n \times n}$ be a Hermitian matrix represented by the list $(T_1,\dots,T_k,C_{1,1}, \dots, C_{n,n})$ as in Definition~\ref{def:rep_matrix_number_field}. 
    Let 
    \begin{equation}
        A = \sum_{j=1}^n \lambda_j \ketbra{\psi_j}{\psi_j},
    \end{equation}
    be the spectral decomposition of $A$ with $(\lambda_j)_{j \in [n]}$ the (real) eigenvalues of $A$ counted with multiplicities, and $(\ket{\psi_j})_{j \in [n]}$ be a suitably chosen orthonormal basis of eigenvectors of $A$. 
    Then, there exists $\mu_j \in \overline{\Q}$ such that $\lambda_j \ketbra{\psi_j}{\psi_j} \in \Q(\theta_1,\dots,\theta_k,\lambda_j, \mu_j)^{n \times n}$. 
    Let $L_j = (T_1,\dots,T_k, \Lambda_j, \mathrm{U}_j, C^{(j)}_{1,1}, \dots, C^{(j)}_{n,n})$ be its representation as in Definition~\ref{def:rep_matrix_number_field}, where $\Lambda_j$ and $\mathrm{U}_j$ are the absolute representations of $\lambda_j$ and $\mu_j$, respectively, and $d = \max_{j \in [k]}d_j$. Then, 
    \begin{equation}\label{eq:complexity_spectral_decomposition}
        \bitlength(L_j) = d^{O(k)} \cdot\poly(\bitlength(A)).
    \end{equation}
\end{proposition}

This proposition is shown in Appendix~\ref{sec_apdx:comput_model}.
As shown in the proof of Theorem~\ref{thm:NP_complete_case}, the channel representation model is only required to efficiently encode the spectral decomposition of Hermitian matrices with algebraic entries.
Consequently, Theorem~\ref{thm:NP_complete_case} holds for any model satisfying Proposition~\ref{prop:rep_spectral_decompo}.

We end this discussion on our representation of channels with the following remarks.
First, we refer the reader to~\cite{Cohen.1993,Gathen.2013} for a more comprehensive background on computational algebraic number theory.
Second, note that in order to practically implement this representation, one might append metadata to the representations in Definitions~\ref{def:rep_matrix_number_field} and~\ref{def:representation_channel} to facilitate memory management or specific operations. 
For instance, it may be useful to store the degree of the generators $(\theta_j)_{j \in [k]}$ or the maximum binary length of the coefficients of their minimal polynomials.
These implementation details do not affect the results of this article, as they introduce a memory overhead only polynomial in the binary length of $\Phi$.

We finally relate $\psucc(\Phi,2)$ to $\eta_{\tr}(\Phi)$ using the Holevo--Helstrom theorem.

\begin{proposition}
\label{prop:optimal_succ_contraction}
Let $\Phi$ be a channel, then
\begin{equation}
\psucc(\Phi,2) = \frac{1}{2} + \frac{1}{2}\eta_{\tr}(\Phi).
\end{equation}
\end{proposition}
\begin{proof}
For any $\rho_1, \rho_2 \in \D(\cH)$, with $\rho_1 \neq \rho_2$, and any 2-outcome POVM $\{M_1, I - M_1\}$,

\begin{align*}
\frac{1}{2} \tr\bigl(M_1 \Phi(\rho_1) + (I - M_1) \Phi(\rho_2)\bigr)
    &= \frac{1}{2} + \frac{1}{2} \tr\bigl(M_1 (\Phi(\rho_1) - \Phi(\rho_2))\bigr) \\
    &\overset{(a)}{\leq} \frac{1}{2} + \frac{1}{4} \| \Phi(\rho_1) - \Phi(\rho_2) \|_1 \\
    &\leq \frac{1}{2} + \frac{1}{2} \frac{\| \Phi(\rho_1) - \Phi(\rho_2) \|_1}{\| \rho_1 - \rho_2 \|_1}
\end{align*}
where we used the Holevo--Helstrom theorem~\cite{Helstrom.1969} for $(a)$. For the other direction, we first use~\cref{lemma:Ruskai} below and choose $\rho_1, \rho_2$ orthogonal, i.e., $\| \rho_1 - \rho_2 \|_1 = 2$, such that $\eta_{\tr}(\Phi) = \frac{1}{2} \| \Phi(\rho_1) - \Phi(\rho_2) \|_1$. Then $M_1$ is the projector onto the positive part of $\Phi(\rho_1) - \Phi(\rho_2)$.
\end{proof}
We summarise in Table~\ref{tab:notations} the notations used in this paper.

\begin{table}[htbp]
    \centering
    \caption{Notations Used Throughout This Paper}
    \begin{tabular}{|p{3.2cm}|p{8.6cm}|p{3.2cm}|}
    \hline
        Symbol & Meaning & Eq.\\
    \hline
        $\cH$ & A finite-dimensional Hilbert space &\\
        $\B(\cH)$ & The set of linear operators on $\cH$ & \\
        $\text{Herm}(\cH)$ & The set of Hermitian operators on $\cH$&\\
        $\D(\cH)$ & The set of density operators on $\cH$& \\
        $I \in \B(\cH)$ & The identity operator in $\B(\cH)$ & \\
        $\text{Id}_n(.)$ & The identity channel on $\B(\mathbb{C}^n)$& \\
        $\rho$, $\sigma$ & Quantum states & \\
        $A,B,\bar{A},\bar{B}$ & Physical systems & \\
        $d_S$ & Dimension of system $S$ &\\
        $\Phi$ & A quantum channel & \\
        $J(\Phi)$ & The Choi state of $\Phi$ & \eqref{eq:choi-state} \\
        $\eta_{\tr}(\Phi)$ & The trace norm contraction coefficient of $\Phi$ & \eqref{eq:contraction_coefficient} \\ $\psucc(\Phi, k)$ & Opt. success prob. of encoding $k$ messages in $\Phi$ & \eqref{eq:proba_success} \\
        $G_{\Phi}$ & The quantum graph associated to $\Phi$ & \eqref{eq:def_qg} \\
        $\|\cdot\|_1$ & The Schatten 1-norm, a.k.a. the trace norm & \\
        $\|\cdot\|_2$ & The Schatten 2-norm, a.k.a. the Frobenius norm & \\
        $\|\cdot\|_\infty$ & The Schatten $\infty$-norm, a.k.a. the operator norm &  \\
        $\|\cdot\|_{\ell_p}$ & The $\ell_p$-norms &  \\
        $S_p^n, p \in [1,\infty]$ & The complex Banach space $(\B(\C^n), \|\cdot\|_p)$ & \\
        $S_p^{H,n}, p \in [1,\infty]$ & The real Banach space $(\text{Herm}(\C^n), \|\cdot\|_p)$ & \\
        $\ell_p^{n}(\mathbb{R}), p \in [1,\infty]$ & The real Banach space $(\mathbb{R}^n, \|\cdot\|_{\ell_p})$ & \\
        $[n]$, $n \in \mathbb{N}$ & The integer set $\{1,\dots,n\}$ & \\
        $\operatorname{conv}\{\mathcal{S}\}$ & The convex hull of the set $\mathcal{S}$ & \eqref{eq:sep}\\
        $E_{i,j}$ & Matrix with $1$ in position $(i,j)$ and $0$ elsewhere   & \\
        $|E|$ & The cardinality of the set $E$ & \\
        $[M]_{\alpha,\beta}$ & The entry in the $\alpha$-th row and $\beta$-th column of the matrix $M$ & \\
        $\gcd(a,b)$ & The greatest common divisor of $a,b \in \Z$ & \eqref{eq:field_extension}, \eqref{eq:field_extension_multiple_generators} \\
        $\overline{\Q}$ & The set of algebraic numbers & \\
        $\bitlength(x)$ & The binary length of $x$ under the representation given by context & \eqref{eq:bit_length_matrix}, \eqref{eq:bit_length_channel} \\
        $\N^*$ & The set of positive integers & \\
    \hline
    \end{tabular}
    \label{tab:notations}
\end{table}

\section{\NP-hardness of computing the trace norm contraction coefficient}\label{sec:NP_general_case}

In this section, we prove \cref{thm:NP_contraction_coef} by reducing the Hermitian non-commutative (resp. real commutative) Little Grothendieck Problem to the problem of computing $\eta_{\tr}(\Phi)$ (resp. computing $\eta_{\tr}(\Phi)$, for $\Phi$ a quantum-classical channel). The complexities of various Grothendieck problems are studied in depth in \cite{briet2015tight} and \cite{ncGrothendieck}.

Before introducing the Hermitian non-commutative Little Grothendieck Problem, let us first reduce the two-parameter optimisation problem which defines $\eta_{\tr}(\Phi)$ to a single-parameter optimisation. Ruskai showed in \cite{Ruskai.1994} the following lemma.

\begin{lemma}[\cite{Ruskai.1994}, Theorem 2; see also Lemma 8.3 of \cite{Wolf}]\label{lemma:Ruskai} For a quantum channel $\Phi : \B(\cH) \rightarrow \B(\cH')$,
\begin{equation}\label{eq:ruskai}
\eta_{\tr}(\Phi) = \max_{\substack{\rho,\sigma \in \D(\cH),\\ \rho \perp \sigma}}\frac{1}{2}\|\Phi(\rho) - \Phi(\sigma)\|_1,
\end{equation}
where $\rho \perp \sigma$ means $\rho$ and $\sigma$ have orthogonal supports. In addition, the maximising states can be taken pure.
\end{lemma}
From this lemma, we can reduce the problem of computing the contraction coefficient to the following optimisation problem over Hermitian operators in the unit ball for the Schatten $\infty$-norm.

\begin{lemma}\label{lemma:dual} For a quantum channel $\Phi : \B(\cH) \rightarrow \B(\cH')$, 
\begin{equation}\label{eq:contraction_coef_vs_diff_eignevalues}
\eta_{\tr}(\Phi) 
= \frac{1}{2}\max_{\substack{X^{*} = X,\\ \|X\|_\infty \leq 1}} (\lambda_{\max} - \lambda_{\min})(\Phi^*(X)),
\end{equation}
where, for any Hermitian operator $O$, $\lambda_{\max}(O)$, respectively $\lambda_{\min}(O)$, is the maximal, respectively minimal, eigenvalue of $O$.
\end{lemma}

\begin{proof}
     In the following equations, the vectors $u$, $v$ are taken of unit length, i.e., $\|u\|_{\ell_2} = \|v\|_{\ell_2} = 1$. By~\cref{lemma:Ruskai},
    \[
    \begin{aligned}
        \eta_{\tr}(\Phi) &= \frac{1}{2} \max_{\substack{u,v \in \cH, \\ u \perp v}} \|\Phi(uu^*) - \Phi(vv^*)\|_1 \\
        &\overset{(a)}{=} \frac{1}{2} \max_{\substack{u,v \in \cH, \\ u \perp v}} \max_{\substack{X^* = X,\\ \|X\|_\infty \leq 1}} | \langle \Phi(uu^* - vv^*), X \rangle| \\
        &= \frac{1}{2} \max_{\substack{u,v \in \cH,\\ u \perp v}} \max_{\substack{X^* = X,\\ \|X\|_\infty \leq 1}} | \langle uu^* - vv^*, \Phi^*(X) \rangle| \\
        &= \frac{1}{2}\max_{\substack{X^* = X,\\ \|X\|_\infty \leq 1}} (\lambda_{\max} - \lambda_{\min})(\Phi^*(X)),
    \end{aligned}
    \]
    where $(a)$ follows from the Hölder duality between the trace and operator norms as well as the fact that $\Phi$ is Hermitian-preserving, and the last step follows directly from writing an eigendecomposition of $\Phi^*(X)$ and using the fact that $u$ and $v$ have norm $1$.
\end{proof}

Hence, as claimed, the optimisation problem of \cref{lemma:dual} is over a single parameter only.

We now switch our focus to the Little Grothendieck Problems, which we will reduce to the computation of the contraction coefficient. A Banach space $X = (\mathcal{X}, \|\cdot\|_X)$ is a linear space $\mathcal{X}$ endowed with a norm $\|\cdot\|_X$ with respect to which $\mathcal{X}$ is complete. For any linear map $\F : X \rightarrow Y$ between two Banach spaces, its operator norm is defined as 
\begin{equation}
\|\F\|_{X\to Y} = \sup_{\|a\|_X \leq 1} \|\F(a)\|_Y,
\end{equation}
i.e., the supremum of $\|\F(a)\|_Y$ when $a$ ranges in the unit ball of $X$. 
The Little Grothendieck Problems are all stated as problems concerning computations of such operator norms of linear maps $\F$ between certain Banach spaces. By abuse of notation, we write $\mathbb{R}^n$ for the Banach space obtained when we endow $\mathbb{R}^n$ with the $\ell_2$ norm.
We write $\ell_1^n(\mathbb{R})$ for the Banach space $\mathbb{R}^n$ endowed with the $\ell_1$ norm and analogously $\ell_\infty^n(\mathbb{R})$ for the Banach space $(\mathbb{R}^n,\|\cdot\|_{\ell_\infty})$. Furthermore, we write $S_1^n= (\B(\C^n), \|\cdot\|_1)$ and $S_\infty^n= (\B(\C^n), \|\cdot\|_\infty)$ for the Banach spaces of complex $n \times n$ matrices endowed with the Schatten $1$- and $\infty$-norm (the former being also called the trace norm and the latter the operator norm). Moreover, we write $S_1^{H,n} = (\text{Herm}(\mathbb{C}^n), \|\cdot\|_1)$ and $S_{\infty}^{H,n} = (\text{Herm}(\mathbb{C}^n), \|\cdot\|_\infty)$ for the \emph{(real)} Banach spaces of $n \times n$ Hermitian complex matrices endowed with, respectively, the Schatten 1- and $\infty$-norm.
The real commutative Little Grothendieck Problem consists in computing $\|\F\|_{\mathbb{R}^n \to \ell_1^d(\mathbb{R})}$ for $\F : \mathbb{R}^n \rightarrow \ell_1^d(\mathbb{R})$ a linear operator. This problem can be generalised to non-commutative settings by taking $\F$ operator-valued. The variant we use for our proof is the Hermitian non-commutative Little Grothendieck Problem, which can be stated as computing $\|\F\|_{\mathbb{R}^n \to S_1^{H,d}}$, given a linear operator $\mathcal{F} : \mathbb{R}^n \rightarrow S_1^{H,d}$ with output space endowed with the Schatten $1$-norm.

Bri\"et \emph{et al.} proved the following hardness theorems for the commutative and Hermitian non-commutative Little Grothendieck Problems in \cite{briet2015tight}.

\begin{theorem}[Theorem 1.3 in \cite{briet2015tight}]\label{thm:hardness_CLGP}
    For any constant $\varepsilon > 0$, it is \texttt{NP}-hard to approximate the \emph{real commutative} Little Grothendieck Problem to within a factor greater than $\sqrt{2/\pi} + \varepsilon$. 
\end{theorem}

\begin{theorem}[Theorem 1.2 and Section 5.1 in \cite{briet2015tight}]\label{thm:hardness_NCLGP}
For any constant $\varepsilon > 0$, it is $\texttt{NP}$-hard to approximate the \emph{Hermitian non-commutative} Little Grothendieck Problem to within a factor greater than $1/\sqrt{2} + \varepsilon$.
\end{theorem}

Finally, note that we can express the Little Grothendieck Problems in terms of the dual of $\F$ as per the argument given in Section 6 of \cite{briet2015tight}. Indeed, for any Banach spaces $X$, $Y$ and any linear map $\F : X \rightarrow Y$, it is a standard fact that the operator norm of $\F$ is equal to that of its dual $\F^* : Y^* \rightarrow X^*$, i.e., $\|\F\|_{X \to Y} = \|\F^*\|_{Y^* \to X^*}$. For the Hermitian non-commutative Little Grothendieck Problem, we have $\F : \mathbb{R}^n \rightarrow S_1^{H,d}$, $\mathbb{R}^n$ is a Hilbert space and thus self-dual and the dual of $S_1^{H,d}$ is $S_\infty^{H,d}$, hence $\F^* : S_{\infty}^{H,d} \rightarrow \mathbb{R}^n$. Therefore, 
\begin{equation}\label{eq:operator_norm_dual}
\|\F\|_{\mathbb{R}^n \to S_{1}^{H,d}} = \|\F^*\|_{S_{\infty}^{H,d} \to \mathbb{R}^n} = \sup_{\substack{\|Y\|_\infty \leq 1,\\Y^*=Y}} \|\F^*(Y)\|_{\ell_2}.
\end{equation}

Since the action of a linear map $\F : \mathbb{R}^n \rightarrow S_1^{H,d}$ is fully determined by where it maps a basis of $\mathbb{R}^n$ to, for all $x \in \mathbb{R}^n$ we can write:

\begin{equation}
\F(x) = \sum_{i=1}^n x_iF_i,
\end{equation}
with, for all $i \in [n]$, $F_i = \F(e_i) \in S_1^{H,d}$, with $\{e_i\mid i \in [n]\}$ the canonical basis of $\R^n$. Then we can express $\F^*$ as:

\begin{equation}
\forall Y \in S_{\infty}^{H,d}: \quad \F^*(Y) = (\langle F_i, Y\rangle)_{i \in [n]}.
\end{equation}

Therefore,

\begin{equation}
\|\F\|_{\mathbb{R}^n \to S_{1}^{H,d}} = \|\F^*\|_{S_{\infty}^{H,d} \to \mathbb{R}^n} = \sup_{\substack{\|Y\|_{\infty} \leq 1, \\Y^*=Y}} \|(\langle F_i, Y \rangle)_{i \in [n]}\|_{\ell_2}.
\end{equation}

The same line of argument also holds \emph{mutatis mutandis} for the real commutative Little Grothendieck Problem. In this case, we have $\F : \mathbb{R}^n \to \ell_1^d(\mathbb{R})$, we write for all $x \in \mathbb{R}^n$ 
\begin{equation}
\F(x) = \sum_{i=1}^n x_if_i,
\end{equation}
with $\{f_i = \F(e_i) \in \mathbb{R}^d \mid i \in [n]\}$. Then, for all $y \in \mathbb{R}^d$,
\begin{equation}
\F^*(y) = (\langle f_i,y \rangle)_{i\in [n]},
\end{equation}
and 
\begin{equation}
\|\F\|_{\mathbb{R}^n \to \ell_1^d(\mathbb{R})} = \|\F^*\|_{\ell_{\infty}^d(\mathbb{R}) \to \mathbb{R}^n} = \sup_{\|y\|_{\ell_\infty} \leq 1}\|\F^*(y)\|_{\ell_2}.
\end{equation}

We now have all the tools to prove \cref{thm:NP_contraction_coef}.

\begin{proof}[Proof of \cref{thm:NP_contraction_coef}]
We start with the hardness result for unital channels as it implies the one for general channels.
Let $\F : \mathbb{R}^n \rightarrow S_1^{H,d}$ be a linear map, and let $\{F_i = \F(e_i) \mid i \in [n]\}$, with $\{e_i \in \mathbb{R}^n \mid i \in [n]\}$ the canonical basis of $\mathbb{R}^n$, so that, for all $x \in \mathbb{R}^n$,
\begin{equation*}
\F(x) = \sum_{i=1}^n x_iF_i,
\end{equation*}
where we assume that the matrices $F_i$ have rational entries, as the hardness result holds even under this restriction (see~\cite[Sections~3.1,~5.1]{briet2015tight}).
Indeed, the reduction of~\cite{briet2015tight} restricts maps with rational entries to subspaces defined by rational linear equations.
Composing these maps with the orthogonal projections onto those subspaces, which also have rational entries, gives maps on the whole Euclidean space with rational entries and the same operator norm.
The proof then proceeds in three steps. First, we construct a unital channel $\Phi_{\alpha}$ out of $\F$, then, we relate the difference $(\lambda_{\max} - \lambda_{\min})(\Phi_{\alpha}^*(Y))$ to $\|\F^*(Y)\|_{\ell_2}$. Finally, we use on the one hand \cref{lemma:dual} to deduce $\eta_{\tr}(\Phi_{\alpha})$ from the difference $(\lambda_{\max} - \lambda_{\min})(\Phi_{\alpha}^*(Y))$ and, on the other hand, that, from \cref{thm:hardness_NCLGP}, it is \NP-hard to maximise $\|\F^*(Y)\|_{\ell_2}$ on the operators $Y\in S^{H,d}_\infty$, with $\|Y\|_\infty \leq 1$, to then conclude that it is also \NP-hard to compute $\eta_{\tr}(\Phi_{\alpha})$.

Thus, let us perform some changes to $\F$ so as to construct a valid quantum channel. Let $k$ be the smallest integer such that $2kd \geq n+1$. 
Let $\widetilde{\F}: \mathbb{R}^n \rightarrow S_1^{H,2kd}$ be the linear map defined as, for all $x \in \mathbb{R}^n$,
\begin{equation}
\widetilde{\F}(x) = \sum_{i=1}^n x_i\oplus_{s=1}^k\bigl(F_i \oplus (-F_i) \bigr) 
= 
\begin{pmatrix}
    \F(x) &  &  & & \\
     & -\F(x) & & \text{\LARGE0} & \\
     & & \ddots & & \\
     &\text{\LARGE0} & & \F(x) & \\
    &&&& - \F(x) \\
\end{pmatrix}.
\end{equation}

Then, $\widetilde{\F}^*(I_{2kd}) = (\langle \oplus_{s=1}^k\bigl(F_i \oplus (-F_i)\bigr), I_{2kd}\rangle)_{i \in [n]} = (k[\langle F_i, I_d \rangle - \langle F_i, I_d \rangle])_{i \in [n]} = 0$. Furthermore, for all $x \in \mathbb{R}^n$, $\|\widetilde{\F}(x)\|_1 = 2k\|\F(x)\|_1$, so that 
\begin{equation}\label{eq:extension}
    \|\F\|_{\mathbb{R}^n \to S_1^{H,d}} = \frac{1}{2k}\|\widetilde{\F}\|_{\mathbb{R}^n \to S_1^{H,2kd}}.
\end{equation}
We write $\{\widetilde{F}_i = \widetilde{\F}(e_i) = \oplus_{s=1}^k\bigl(F_i \oplus (-F_i)\bigr) \mid i \in [n]\}$. 
Let $\Psi : \B(\mathbb{C}^{2kd}) \rightarrow \B(\mathbb{C}^{n+1}) \subseteq \B(\mathbb{C}^{2kd})$ be defined as, for all $Y \in \B(\mathbb{C}^{2kd})$:
\begin{equation}\label{eq:channel_out_of_LGP}
    \Psi(Y) = \bigg(\sum_{i=1}^n \langle \widetilde{F}_i,Y \rangle e_i\bigg)e_{n+1}^T+e_{n+1}\bigg(\sum_{i=1}^n \langle \widetilde{F}_i,Y \rangle e_i\bigg)^T = \widetilde{\F}^*(Y)e_{n+1}^T + e_{n+1}\widetilde{\F}^*(Y)^T.
\end{equation}
Note that $\Psi$ is linear over $\mathbb{C}$, so that we can turn it into the dual of a unital channel $\Phi$ by shifting and renormalising.
Let $\Phi^*_\alpha$ be defined as: 
\begin{equation}
    \forall Y \in \B(\mathbb{C}^{2kd}), \quad \Phi_\alpha^*(Y) := \alpha \Psi(Y) + \frac{\tr(Y)}{2kd}I_{2kd},
\end{equation}
with $\alpha > 0$. For $\alpha$ small enough, $\Phi_\alpha^*$ is completely positive. Indeed, the Choi operator $J_\alpha$ of $\Phi_\alpha^*$, defined in \eqref{eq:choi-state}, is 
\begin{equation*}
    J_\alpha = \alpha J_\Psi + \frac{1}{(2kd)^2}I_{(2kd)^2},
\end{equation*}
with $J_\Psi$ and $\tfrac{1}{(2kd)^2}I_{(2kd)^2}$ the respective Choi operators of the linear maps $\Psi$ and $Y \mapsto \tfrac{\tr(Y)}{2kd}I_{2kd}$. Then, for $\alpha$ small enough, $J_\alpha$ is positive semidefinite. As a map is completely positive if and only if its Choi state is positive semidefinite (cf. \cite[Theorem~2.22]{watrous.2018}), $\Phi_\alpha^*$ is completely positive for small enough $\alpha$. Note that such an $\alpha$ can be chosen rational and computed in polynomial time from the $\widetilde{F}_i$, $i \in [n]$. 
The Choi matrix $J(\Phi_\alpha)$ is then also rational.
Exact rational symmetric Gaussian elimination provides a decomposition of $(2kd)J(\Phi_\alpha)$ into at most $(2kd)^2$ positive rank-one matrices with rational entries, in polynomial time. 
This yields Kraus operators, each being a rational matrix multiplied by the square root of a positive rational number.
Their representations under Definition~\ref{def:representation_channel} are therefore computable in polynomial time and have polynomial binary length.
Furthermore, as $\widetilde{\F}^*(I_{2kd}) = 0$, $\Phi_\alpha^*$ is unital, so that, for $\alpha$ small enough, $\Phi_\alpha$ is a channel. Moreover, it is easy to see that the output of $\Psi$ is always traceless, so that $\Phi^*_\alpha$ is trace preserving. Thus, its dual $\Phi_\alpha$ is also unital.
Finally, remark that, for $Y \in \text{Herm}(\mathbb{C}^{2kd})$, 
\begin{align}
    (\lambda_{\max}- \lambda_{\min})(\Phi_{\alpha}^{*}(Y)) &= (\lambda_{\max}- \lambda_{\min})(\alpha \Psi(Y)) \\
    &\overset{(a)}{=} 2\alpha \|(\langle \widetilde{F}_i,Y\rangle)_{i \in [n]}\|_{\ell_2}\label{eq:lambda_to_norm} \\
    &= 2\alpha \|\widetilde{\F}^*(Y)\|_{\ell_2},\label{eq:2norm_vs_difference_eigenvalues}
\end{align}
where $(a)$ follows from the fact that any operator $A = be^* + eb^*$ with $e \perp b$ satisfies $\lambda_{\max}(A) = - \lambda_{\min}(A) = \|b\|_{\ell_2}\|e\|_{\ell_2}$. Note that because $\widetilde{F}_i$ and $Y$ are Hermitian, the vector $\sum_{i=1}^n \langle \widetilde{F}_i, Y \rangle e_i$ is real so its transpose and complex conjugate are the same. Moreover, writing $Y = (Y_{ab})_{a,b\in[2k]}$ with $Y_{ab} \in \B(\mathbb{C}^d)$, and similarly for $\widetilde{F}_i$, the fact that $\widetilde{F}_i$ is block diagonal, with diagonal blocks $(\widetilde{F}_i)_{2s-1,2s-1} = F_i$ and $(\widetilde{F}_i)_{2s,2s} = -F_i$, gives, for all $i \in [n]$,
\begin{equation*}
\langle \widetilde{F}_i, Y \rangle = \sum_{a=1}^{2k} \langle (\widetilde{F}_i)_{aa}, Y_{aa} \rangle = \Big\langle F_i, \sum_{s=1}^k \big(Y_{2s-1,2s-1} - Y_{2s,2s}\big) \Big\rangle,
\end{equation*}
that is, $\widetilde{\F}^*(Y) = \F^*\big(\sum_{s=1}^k (Y_{2s-1,2s-1} - Y_{2s,2s})\big)$.
Since $Y_{aa}$ is Hermitian with $\|Y_{aa}\|_\infty \leq \|Y\|_\infty$, the triangle inequality shows that the supremum of $\|\widetilde{\F}^*(Y)\|_{\ell_2}$ over Hermitian $Y \in \B(\C^{2kd})$ with $\|Y\|_\infty \leq 1$ is at most $2k$ times the supremum of $\|\F^*(Y')\|_{\ell_2}$ over Hermitian $Y' \in \B(\C^d)$ with $\|Y'\|_\infty \leq 1$. Conversely, for any Hermitian $Y'  \in \B(\C^d)$ with $\|Y'\|_\infty \leq 1$, the choice $Y = \oplus_{s=1}^k (Y' \oplus -Y')  \in \B(\C^{2kd})$ satisfies $\|Y\|_\infty = \|Y'\|_\infty$ and $\widetilde{\F}^*(Y) = 2k\,\F^*(Y')$. Hence, by \eqref{eq:2norm_vs_difference_eigenvalues},
\begin{equation}\label{eq:sup_block_reduction}
\sup_{\substack{Y^*=Y,\\ \|Y\|_\infty \leq 1}} (\lambda_{\max} - \lambda_{\min})(\Phi_\alpha^*(Y)) = 4k\alpha \sup_{\substack{Y'^*=Y',\\ \|Y'\|_\infty \leq 1}} \|\F^*(Y')\|_{\ell_2}.
\end{equation}
Thus, combining \eqref{eq:sup_block_reduction} with both \eqref{eq:operator_norm_dual} and \cref{lemma:dual}, we have that:

\begin{equation}
\|\F\|_{\mathbb{R}^n\to S_1^{H,d}} = \frac{1}{2k\alpha}\, \eta_{\tr}(\Phi_{\alpha}),
\end{equation}
which ends the reduction and shows by \cref{thm:hardness_NCLGP} that it is \NP-hard to approximate $\eta_{\tr}(\Phi_\alpha)$ to within a factor $1/\sqrt{2} + \varepsilon$ for $\varepsilon > 0$.

Let us now move to the second point of Theorem~\ref{thm:NP_contraction_coef} about quantum-classical channels. The proof uses exactly the same techniques as the proof of the first point. 
In fact, given $\F : \R^{n} \to \ell_1^d(\R)$, we write as before $\F(e_i) = f_i \in \R^d$ for $\{e_i\}_{i \in [n]}$ the canonical basis of $\R^n$. Then, 
\begin{align*}
\F(x) = \sum_{i=1}^nx_if_i,\qquad \F^*(y) = \sum_{i=1}^n \langle f_i,y \rangle e_i,
\end{align*}
with $\F^* : \ell_{\infty}^d(\R) \rightarrow \R^n$ the dual of $\F$. Similarly to the general quantum case, we set $\tilde{f_i} = \oplus_{s=1}^k\bigl(f_i \oplus (-f_i)\bigr) \in \R^{2kd}$ with $k$ the smallest integer such that $2kd \geq n+1$ and $\tilde{F_i} = \operatorname{diag}(\tilde{f_i})$ the diagonal $2kd \times 2kd$ matrix whose diagonal entries are the entries of $\tilde{f_i}$. We then define the map $\Psi : \B(\C^{2kd}) \rightarrow \B(\C^{n+1}) \subseteq \B(\C^{2kd})$, for all $Y \in  S^{2kd}_{\infty}$:
\begin{equation}
\Psi(Y) = \left(\sum_{i=1}^n \langle \tilde{F}_i, Y \rangle e_i\right)e_{n+1}^T +e_{n+1}\left(\sum_{i=1}^n \langle \tilde{F}_i, Y \rangle e_i\right)^T.
\end{equation}
Then, for $\alpha > 0$ small enough, the map $\Phi^*_{\alpha}$ defined as
\begin{equation}
\forall Y \in \B(\C^{2kd}), \quad \Phi^*_{\alpha}(Y) = \alpha \Psi(Y) + \frac{\tr(Y)}{2kd}I_{2kd},
\end{equation}
is unital completely positive and trace preserving, so that $\Phi_\alpha$ is a unital channel. Note that, as the $\tilde{F}_i$ are diagonal, $\Phi_\alpha$ is a quantum-classical channel, i.e., its output is diagonal for all inputs. Finally, using \cref{lemma:dual},
\begin{align}
\eta_{\tr}(\Phi_\alpha) &= \frac{1}{2}\sup_{\substack{Y^* = Y, \\\|Y\|_\infty \leq 1}}(\lambda_{\max} - \lambda_{\min})(\Phi^*_\alpha(Y)) \\
&= \frac{\alpha}{2}\sup_{\substack{Y^* = Y, \\ \|Y\|_\infty \leq 1}}(\lambda_{\max} - \lambda_{\min})(\Psi(Y)) \\
&\overset{(a)}{=}\alpha \sup_{\substack{Y^* = Y, \\ \|Y\|_\infty \leq 1}} \|(\langle \tilde{F}_i, Y\rangle)_{i \in [n]}\|_{\ell_2} \label{eq:reviewer_a} \\
&\overset{(b)}{=} \alpha \sup_{\substack{y \in \R^{2kd},\\ \|y\|_{\ell_\infty} \leq 1}} \|(\langle \tilde{f}_i, y\rangle )_{i\in [n]}\|_{\ell_2} \label{eq:reviewer_b} \\
&= 2k\alpha \sup_{\substack{y \in \R^{d}, \\ \|y\|_{\ell_\infty} \leq 1}} \|(\langle f_i, y \rangle)_{i\in [n]}\|_{\ell_2} \\
&= 2k\alpha\|\F^*\|_{\ell_\infty^{d}(\mathbb{R}) \to \mathbb{R}^n} \\
&= 2k\alpha \|\F\|_{\mathbb{R}^n \to \ell_1^{d}(\mathbb{R})},
\end{align}
where $(a)$ follows by the same argument as \eqref{eq:lambda_to_norm}, and $(b)$ by the fact that $\tilde{F_i} = \operatorname{diag}(\tilde{f_i})$.

Finally by \cref{thm:hardness_CLGP} it is \NP-hard to approximate $\|\F\|_{\mathbb{R}^n \to \ell_1^d(\mathbb{R})}$ to within a factor greater than $\sqrt{2/ \pi} + \varepsilon$ for all $\varepsilon > 0$, from which we conclude that it is \NP-hard to approximate $\eta_{\tr}(\Phi_\alpha)$ to within the same factor.
The same argument as above shows that the Kraus representation of $\Phi_\alpha$ can also be constructed in polynomial time under Definition~\ref{def:representation_channel}.
\end{proof}

\section{\NP-hardness of the complete case}\label{sec:NP_hardness_complete_case}

In the previous section, we showed that it is \NP-hard to approximate the contraction coefficient of a channel. In this section, we strengthen this result, showing it is \NP-hard even in the ``complete'' case where we want to test whether a channel has contraction coefficient $1$. 

For the proof, we use several characterisations of channels with $\eta_{\tr}(\Phi) = 1$: one in terms of optimisation over separable states and the other in terms of the confusability graph of $\Phi$. We start by defining the set of separable states.

\begin{definition}[Separable states]
    \begin{equation}\label{eq:sep} 
    \begin{split}
    \SEP(\cH \otimes \cH') &\coloneqq \conv \{ \rho \otimes \sigma \mid \rho \in \D(\cH), \sigma \in \D(\cH')\} \\
    &= \conv\{ uu^{*} \otimes v v^{*} \mid u \in \cH, v \in \cH', \|u\|_{\ell_2}=\|v\|_{\ell_2}=1\}.
    \end{split}
    \end{equation} 
    For Hermitian $A \in \B(\cH \otimes \cH')$, 
    \begin{equation} 
    \hsep(A) := \min_{X \in \SEP(\cH \otimes \cH')} \langle A, X \rangle = \min_{\substack{u \in \cH, v \in \cH',\,\\ \|u\|_{\ell_2}=\|v\|_{\ell_2}=1}} \langle A, u u^{*} \otimes v v^{*} \rangle .
    \end{equation}
\end{definition}

Now, we introduce the confusability graph of a channel.
The \emph{confusability graph} of a classical channel was introduced by Shannon in \cite{Shannon.1956} to study the zero-error  capacity of the channel.
This graph has the set of inputs of the channel as its vertex set and two inputs are connected by an edge if and only if there is a non-zero probability that the channel maps them to the same output. This structure was generalised to quantum channels in~\cite{duan2009super}. 

\begin{definition}[Quantum confusability graph, (3) in~\cite{duan2009super}]
For $\Phi : \B(\cH) \rightarrow \B(\cH')$ defined by $\Phi(x) = \sum_i K_ixK_i^*$, for all $x \in \B(\cH)$, its \emph{quantum confusability graph} is
\begin{equation}\label{eq:def_qg}
G_{\Phi} = \operatorname{Span}\{K_i^*K_j \mid i,j\}.
\end{equation}
\end{definition}
It is simple to see that $G_{\Phi}$ does not depend on the specific choice of Kraus operators and forms an operator system, i.e., a unital self-adjoint subspace of $\mathcal{B}(\mathcal{H})$. In fact, any operator system is the confusability graph of some channel $\Phi$:
\begin{lemma}[Lemma 2 in \cite{duan2009super}]\label{lemma:quantum_graph}\label{lem:subspace-to-channel}
    For a subspace $S \subseteq \B(\mathbb{C}^n)$, there is a channel $\Phi$ such that $G_{\Phi} = S$ if and only if $S$ is self-adjoint (i.e., if $x \in S$ then $x^* \in S$) and contains the identity.

    Furthermore, let $S$ be self-adjoint and contain the identity, let $d = \dim(S)$, and let $\{M_1,\dots,M_d\}$ be a Hermitian orthonormal basis of $S$ for the Hilbert--Schmidt inner product. Set $M_{d+i} := -M_i$ for $i \in [d]$. Then there is a channel $\Phi : \B(\C^n) \to \B(\C^{n'})$ with $n'$ polynomial in $n$ with Kraus operators $\{K_i \mid i \in [2d]\}$, satisfying $G_{\Phi} = S$ and, for all $i,j \in [2d]$, $K_i^*K_j = \delta_{i,j}P_i$, where
    \begin{equation*}
        P_i := \frac{1}{2d}\,(I + M_i) \succeq 0.
    \end{equation*}
    Finally, for every $x \in \B(\C^n)$ with $\tr(x) = 0$,
    \begin{equation*}
        \sum_{i,j} \bigl|\langle K_i^*K_j, x\rangle\bigr|^2 \;\geq\; \frac{1}{2d^2}\,\|\Pi_S(x)\|_2^2,
    \end{equation*}
    where $\Pi_S$ denotes the orthogonal projection onto $S$.

\end{lemma}

\begin{remark}\label{rem:proof_graph_to_channel}
    The second paragraph of the previous lemma is not in the statement of Lemma 2 of \cite{duan2009super} but follows its proof; we present here a variation of that proof, which also gives the quantitative bound stated at the end of the lemma.

    Since $S = S^*$, every $x \in S$ can be written as $x = h_1 + i h_2$, where $h_1 = \frac{1}{2}(x + x^*)$ and $h_2 = \frac{1}{2i}(x - x^*)$ are Hermitian elements of $S$. Hence $S$ is spanned by its Hermitian elements, and we can extract from them a basis of $S$. As the Hilbert--Schmidt inner product of two Hermitian matrices is real, Gram--Schmidt orthonormalisation of this basis only involves real coefficients, and therefore gives a Hermitian orthonormal basis $M_1,\dots,M_d$ of $S$. Note that $\|M_i\|_\infty \leq \|M_i\|_2 = 1$, hence $I + M_i \geq 0$ and $P_i \geq 0$ for all $i \in [2d]$. Moreover, since $M_{d+i} = -M_i$,
    \begin{equation*}
        \sum_{i=1}^{2d} P_i = I + \frac{1}{2d}\sum_{i=1}^{2d} M_i = I.
    \end{equation*}
    We can diagonalise each $P_i$ as
    \begin{equation}
        P_i = \sum_{j=1}^{r_i} p_j^{(i)} \proj{\psi_j^{(i)}},
    \end{equation}
    where $r_i$ is the rank of $P_i$ and $p_j^{(i)} > 0$ are its non-zero eigenvalues repeated according to their multiplicity. Let $n' = \sum_{i=1}^{2d} r_i \leq 2dn$. Let $\{\ket{j^{(i)}} \mid j \in [r_i]\}$ be the computational basis of $\C^{r_i}$ for $i \in [2d]$, so that $\C^{n'} = \bigoplus_{i=1}^{2d} \C^{r_i}$, and let
    \begin{equation}\label{eq:kraus_from_basis}
        K_i = \sum_{j=1}^{r_i} \sqrt{p_j^{(i)}} \ketbra{j^{(i)}}{\psi_j^{(i)}}.
    \end{equation}
    Then $K_i^*K_j = \delta_{i,j}P_i$, because the $K_i$ have orthogonal ranges, and since $\sum_{i}P_i = I$ the map $\Phi : x \mapsto \sum_{i=1}^{2d} K_i x K_i^*$ is a quantum channel. Its quantum confusability graph is $G_\Phi = \operatorname{Span}\{P_i \mid i \in [2d]\}$. This space contains $P_i + P_{d+i} = \frac{1}{d}I$ and $P_i - P_{d+i} = \frac{1}{d}M_i$ for all $i \in [d]$, hence it contains $S$ since $\{M_i\}$ is a basis of $S$; conversely each $P_i$ belongs to $S$ since $I \in S$. Therefore $G_\Phi = S$.

    For the last statement, let $x$ with $\tr(x) = \langle I, x\rangle = 0$. Then
    \begin{equation*}
        \sum_{i,j}|\langle K_i^*K_j, x\rangle|^2 \geq \sum_{i=1}^{2d}|\langle P_i, x\rangle|^2 = \frac{1}{4d^2}\sum_{i=1}^{2d}|\langle M_i, x\rangle|^2 = \frac{1}{2d^2}\sum_{i=1}^{d}|\langle M_i, x\rangle|^2 = \frac{1}{2d^2}\|\Pi_S(x)\|_2^2.
    \end{equation*}
\end{remark}

For the operator systems arising in the reduction of Theorem~\ref{thm:NP_complete_case}, this construction can be implemented in polynomial time in the size of the Graph 2-CSP instance. 
Indeed, as shown in Appendix~\ref{sec_apdx:details}, each $P_i$ has a polynomial-size representation over a number field of degree at most eight.
The exact spectral computations used to construct the Kraus operators can therefore be performed in polynomial time.

\begin{remark}
Note also that given $S \subseteq \B(\C^n)$ self-adjoint containing the identity, a channel $\Phi$ such that $G_{\Phi} = S$ and $K_i^*K_j = \delta_{i,j}P_i$ with $P_i \succeq 0$, as in the second part of the previous lemma, is unital if and only if every $P_i$ is an orthogonal projector (indeed $K_iK_i^*$ acts on the $i$-th block of the output space with the non-zero eigenvalues of $P_i$). In that case $G_\Phi = \operatorname{Span}\{P_i\}$ is a commutative algebra, so such a channel can \emph{not} be chosen unital unless $S$ is of this form (for instance $S = \C I_n$, for which one can trivially take $\Phi = \text{\normalfont Id}_n(\cdot)$). Therefore, although we have shown in \cref{thm:NP_contraction_coef} that it is \NP-hard to approximate $\eta_{\tr}(\Phi)$ for $\Phi$ promised to be unital, we do not show that it is hard to decide whether $\eta_{\tr}(\Phi) = 1$ for such a channel. 
\end{remark}

\begin{theorem} \label{thm:contraction_coeff_equal_1}
    Let $\Phi : \B(\cH) \to \B(\cH')$ be a quantum channel, then the following are equivalent
    \begin{enumerate}
        \item $\eta_{\tr}(\Phi) = 1$;
        \item there exists a rank-one matrix $x$ such that $x \in G_{\Phi}^{\perp}$;
        \item $\hsep(J(\Phi^{*} \circ \Phi)) = 0$.
    \end{enumerate}
\end{theorem}

\begin{proof}
    Items $1)$ and $2)$ follow from \cite[Lemma~1]{duan2009super}. For item $3)$, note that by~\cite[Proposition 2.20]{watrous.2018}
    \begin{equation}
        J(\Phi^* \circ \Phi) = \frac{1}{\dim(\cH)}\sum_{i,j} \Wvec((K_i^*K_j)^T) \Wvec((K_i^*K_j)^T)^*, 
    \end{equation}
and that $\Wvec : \B(\cH) \to \cH\otimes \cH$, introduced in \eqref{eq:def_vec}, is an isomorphism. Hence, if there is a rank-one matrix $x = uv^* \in G_\Phi^\perp$, then
\begin{equation}
    \langle J(\Phi^*\circ \Phi), \bar u\bar u^*\otimes vv^* \rangle = 0.
\end{equation}
Conversely, if $X = \sum_a \lambda_a u_au_a^* \otimes v_av_a^* \in \SEP$ is such that $\langle J(\Phi^* \circ \Phi),X\rangle = 0$, then for all $a,i,j$, 
\begin{equation}
    \langle \Wvec((K_i^*K_j)^T), u_a\otimes v_a \rangle = 0,
\end{equation}
so that $\bar u_a v_a^* \in G_\Phi^\perp$.
\end{proof}
See \cite[Theorem A.11]{Singh.2024} for a summary of related characterisations.

We now introduce the \NP-hard problem that we will embed into a channel $\Phi$ such that ``yes'' instances are mapped to channels satisfying the equivalent properties in \cref{thm:contraction_coeff_equal_1} and ``no'' instances are mapped to channels that do not. The problem we use is Graph $2$-CSP.

\begin{definition}[Graph 2-CSP, see Definition 3 in \cite{LeGall.2012}]\label{def:Graph_2_CSP}
A \emph{constraint graph} $G = (V, E)$ is an undirected graph along with a set $\Sigma$ of ``colours'' and, for each edge $e = (v,u) \in E$, a mapping $R_e : \Sigma \times \Sigma \rightarrow \{0,1\}$ called the constraint on $e$. A mapping $\tau : V \rightarrow \Sigma$ (called a \emph{colouring}) satisfies the constraint $R_e$ if $R_e(\tau(v), \tau(u)) = 1$ for an edge $e = (v,u)$ in $E$. The graph $G$ is said to be \emph{satisfiable} if there is a colouring $\tau$ that satisfies all the constraints, while $G$ is said to be $(1-\eta)$-unsatisfiable if for all colourings $\tau$, the fraction of constraints satisfied by $\tau$ is at most $1-\eta$.
\end{definition}

\begin{theorem} \label{thm:2CSP_Hardness}
    There is a universal constant $\gamma > 0$ such that, given a Graph 2-CSP instance, constructed as in~\cite[Theorem~3]{LeGall.2012},\footnote{Theorem~3 in~\cite{LeGall.2012} provides a specific family of Graph 2-CSP instances for which~\cite[Theorem~5]{LeGall.2012} constructs a quantum verifier.} and promised to be either (1) satisfiable, or (2) $(1-\gamma)$-unsatisfiable, it is \NP-hard to distinguish between the two cases. 
\end{theorem}
\begin{proof}
    Follows from Theorem 3 of \cite{LeGall.2012}.
\end{proof}

In order to embed Graph $2$-CSP into the confusability graph of a quantum channel, we will follow the ideas of \cite{BlierTapp} and \cite{LeGall.2012} relating to languages with short \emph{quantum proofs}. 

\begin{definition}\label{def:QMA_log}
    A promise problem $(L,\overline{L})$ (i.e., $L,\overline{L}$ disjoint subsets of $\{0,1\}^*$) is in $\QMA_{\log}(2,a,b)$ if there is a polynomial $p(m)$ and a polynomial-time classical verification algorithm $\mathcal{V}$ that on input $x$ of length $m$ prepares a quantum circuit $\cV(x)$ acting on $O(\log m)$ qubits such that for any $m$ and any instance of size $m$, we have
    \begin{itemize}
        \item If $x \in L$, there are $u, v \in \mathbb{C}^{p(m)}$ unit vectors such that 
        \begin{equation}
        \tr(\A_{\cV(x)} uu^* \otimes vv^*) \geq a
        \end{equation}
        where $\A_{\cV(x)}$ is the acceptance effect of the verification circuit $\cV(x)$, i.e., $\A_{\cV(x)}$ is a positive semidefinite operator such that applying $\cV(x)$ on its input qubits and then reading its measurement output is equivalent to measuring the input with the POVM $\{\A_{\cV(x)}, I - \A_{\cV(x)}\}$.
        \item If $x \in \overline{L}$, for any unit vectors $u, v \in \mathbb{C}^{p(m)}$, we have
        \begin{equation}
        \tr(\A_{\cV(x)}  uu^* \otimes vv^*) \leq b. \end{equation}
    \end{itemize}
\end{definition}

    Theorem~5 of \cite{LeGall.2012} relates Graph $2$-CSP to the above class.

    \begin{theorem}
    \label{thm:2csp-qma}
        For the constant $\gamma > 0$ given in \cref{thm:2CSP_Hardness},
        let $L$ and $\overline{L}$ be, respectively, the sets of satisfiable and of $(1-\gamma)$-unsatisfiable Graph $2$-CSP instances among the constraint graphs constructed as in~\cite[Theorem~3]{LeGall.2012}. Then $(L,\overline{L}) \in \QMA_{\log}(2,1,1-\Omega(1/p(m)))$. In addition, the protocol is such that $p(m) = O(m)$. 
    \end{theorem}

    We now describe the steps of the reduction from Graph 2-CSP to the contraction coefficient problem. For improved readability, we state a lemma for each step of the reduction. The first step is, given a Graph 2-CSP instance $G$, to construct a projector $\Pi$.
    \begin{lemma}[From $G$ to $\Pi$]
    \label{cor:hardness_sep}
        Given a Graph 2-CSP instance $G$ as in \cref{thm:2csp-qma}, we can construct in polynomial time a projector $\Pi \in \B(\cH \otimes \cH)$ with $\dim \cH = n$ such that if $G$ is satisfiable then $\hsep(\Pi) = 0$ and if $G$ is $(1-\gamma)$-unsatisfiable, then $\hsep(\Pi) = \Omega(\frac{1}{n})$.
    \end{lemma}
    \begin{proof}
    Let $\cV$ be the verifier coming from \cref{thm:2csp-qma}. Consider an instance $G$ of Graph 2-CSP of size $m$ and 
    we let $\Pi$ be the orthogonal projector on the orthogonal complement of the kernel of $I - \A_{\cV(G)}$, where $\A_{\cV(G)} \in \B(\cH \otimes \cH)$ is the acceptance effect on input $G$.
    Let $\cH = \C^{p(m)}$ and $n := p(m) = O(m)$. 
    Then we know that if $G$ is satisfiable, there exist unit vectors $u,v \in \cH$ such that $\tr(\A_{\cV(G)}\, uu^* \otimes vv^*) = 1$ so, as $I - \A_{\cV(G)} \succeq 0$, $u \otimes v \in \ker(I - \A_{\cV(G)}) = \ker \Pi$ and $\tr(\Pi\, uu^* \otimes vv^*) = 0$. As a result, $\hsep(\Pi) = 0$. If $G$ is $(1-\gamma)$-unsatisfiable, then for any unit vectors $u,v \in \cH$, $\tr(\A_{\cV(G)}\, uu^* \otimes vv^*) \leq 1 - \Omega(\frac{1}{n})$, i.e., $\tr((I - \A_{\cV(G)})\, uu^* \otimes vv^*) = \Omega(\frac{1}{n})$. Since $\Pi$ is the orthogonal projector onto the image of $I - \A_{\cV(G)}$ and $0 \preceq I - \A_{\cV(G)} \preceq I$, we have $I - \A_{\cV(G)} \preceq \Pi$, so that $\tr(\Pi\, uu^* \otimes vv^*) = \Omega(\frac{1}{n})$ for all unit vectors $u,v$, i.e., $\hsep(\Pi) = \Omega(\frac{1}{n})$. In addition, note that as $\cV$ is a polynomial-time classical algorithm and the circuit $\cV(G)$ acts on $O(\log m)$ qubits, $\Pi$ can be computed in polynomial time, see Appendix~\ref{sec_apdx:details}.
    \end{proof}
    
    The next step is to construct a subspace $S$ from $\Pi$ such that $\hsep(\Pi)$ is related to the distance between rank-one matrices and $S$.
    
For this purpose, we use the correspondence $\Wvec$ between $\C^{p \times q}$ and $\C^p \otimes \C^q$ introduced in \eqref{eq:def_vec}.
Note that an operator $x \in \mathbb{C}^{p \times q}$ is rank one if and only if $x = uv^*$ with $u \in \mathbb{C}^p\backslash\{0\}$, $v \in \mathbb{C}^q\backslash\{0\}$ and thus if and only if $\Wvec(x) = u \otimes \bar{v}$ with $\bar{v}$ the complex conjugate of $v$ in the basis $\{\ket{j} \mid j \in [q]\}$. 

    \begin{lemma}[From $\Pi$ to $S$]
    \label{lem:from-pi-to-s}
        Let $\Pi \in \B(\cH \otimes \cH)$ be an orthogonal projection and $K$ be the kernel of $\Pi$. We construct the subspace $S$ of $\B(\cH)$ as $S := \Wvec^{-1}(K)$. We have 
        \begin{equation}
        \langle \Pi, u u^{*} \otimes v v^{*} \rangle = \operatorname{dist}(u \otimes v, K)^{2} = \operatorname{dist}(u v^{T}, S)^{2} ,
        \end{equation}
with
\begin{equation}
\operatorname{dist}(x,K) := \min_{y \in K} \|x-y\|_{\ell_2}, \qquad \operatorname{dist}(x,S) := \min_{y \in S} \|x - y\|_2.
\end{equation}
\end{lemma}

\begin{proof}
    Let $\{a_i \in \cH \otimes \cH \mid i\}$ be an orthonormal basis of the support of $\Pi$ so that $\Pi = \sum_{i} a_i a_i^*$. Then
    \begin{align*}
        \langle \Pi, u u^{*} \otimes v v^{*} \rangle
        &= \sum_{i} \langle a_ia_i^* , u u^{*} \otimes v v^{*} \rangle \\
        &= \sum_{i}  |\langle a_i, u \otimes v \rangle|^2 \\
        &= \dist(u \otimes v, K)^2 \\
        &= \dist(\Wvec^{-1}(u \otimes v), \Wvec^{-1}(K))^2,
    \end{align*}
    which gives the desired result. 
\end{proof}

The constructed subspace $S$ has an orthogonal complement that is not necessarily the confusability graph of a quantum channel. But we can construct a subspace $\widehat{S}$ that has this property and that behaves in the same way as $S$.

\begin{lemma}[From $S$ to $\widehat{S} = G_{\Phi}^{\perp}$]
\label{lemma:space_extension}
    Let $S \subseteq \mathbb{C}^{n \times m}$ be a subspace, then 
        \begin{equation}
        \widehat{S} := E_{01}\otimes S + E_{10}\otimes S^* = \left\{ \begin{pmatrix} 0 & A \\ B^* & 0\end{pmatrix} \;\middle|\; A,B \in S\right\} \subseteq \B(\mathbb{C}^{n+m})
        \end{equation}
    satisfies $(a)~\widehat{S}^{\perp} = \widehat{S}^{\perp *}$, i.e., $\widehat{S}^\perp$ is self-adjoint, $(b)~ I_{n+m} \in \widehat{S}^{\perp}$ and $(c)~$ there is a constant $C > 0$ such that
    \begin{equation}
    \label{eq:minrank-s-hats}
    \inf_{\substack{x \in \B(\C^{n+m})  \rank(x) = 1,\\ \|x\|_2 = 1}} \dist(x, \widehat{S}) \leq \inf_{\substack{x \in \C^{n \times m} \\ \rank(x) = 1,\\ \|x\|_2 = 1}} \dist(x, S) \leq 
    C \inf_{\substack{x \in \B(\C^{n+m}) \\ \rank(x) = 1, \\\|x\|_2 = 1}} \dist(x, \widehat{S}).  
    \end{equation}    
    Furthermore, let $(s^\perp_1, \dots,s^\perp_d)$ be an orthonormal basis of $S^\perp$, with $d = \dim(S^{\perp})$, 
    let $(h^{(1)}_1,\dots,h_{n^2}^{(1)})$ and $(h^{(2)}_1,\dots,h^{(2)}_{m^2})$ be, respectively, a Hermitian orthonormal basis of $\B(\mathbb{C}^n)$ and $\B(\mathbb{C}^m)$, let
   \begin{align}
    B_1 &= \left\{ \begin{pmatrix} h^{(1)}_p & 0 \\ 0 & 0 \end{pmatrix} \;\middle|\; p \in [n^2]\right\} \\
    B_2 &= \left\{\begin{pmatrix} 0 & 0 \\ 0 & h_p^{(2)} \end{pmatrix} \;\middle|\; p \in [m^2]\right\}\\
    B_3 &= \left\{ \frac{1}{\sqrt{2}}\begin{pmatrix} 0 & s^\perp_p \\ s^{\perp *}_{p} & 0 \end{pmatrix} \;\middle|\;  p \in [d]\right\} \\
    B_4 &= \left\{\frac{1}{\sqrt{2}} \begin{pmatrix} 0 & i s_p^\perp \\ -is^{\perp *}_p & 0  \end{pmatrix} \;\middle|\; p \in [d]\right\}.
    \end{align}
    Then, 
    $B_{\widehat{S}^\perp} = B_1 \cup B_2 \cup B_3 \cup B_4$
    is a Hermitian orthonormal basis of $\widehat{S}^\perp$.
\end{lemma}

\begin{proof}
Points $(a)$ and $(b)$ are easy. For $(c)$, we have
    \begin{equation*}
    \dist\left(\begin{pmatrix} 0 & x \\ 0 & 0\end{pmatrix}, \widehat{S} \right) = \dist(x,S),
    \end{equation*}
    and this proves that the first inequality in~\eqref{eq:minrank-s-hats} holds. 
    For the second one, let 
    \begin{equation}
        \begin{pmatrix}
            a & b \\
            c & d
        \end{pmatrix} \in \B(\mathbb{C}^{n +m}) \text{ with } \Bigl\| \begin{pmatrix} a & b \\ c & d \end{pmatrix} \Bigr\|_2 = 1,
    \end{equation}
    and let 
    \begin{equation}
    \delta^2 :=
    \dist \left( \begin{pmatrix} a & b \\ c & d\end{pmatrix}, \widehat{S}\right)^2 = \|a\|_2^2 + \|d\|_2^2 + \dist(b, S)^2 + \dist(c,S^*)^2.
    \end{equation}
    We will show that there exists $x \in \C^{n \times m}$ with $\| x \|_2 = 1$ such that $\rank(x) \leq \rank\left(\begin{pmatrix} a & b \\ c & d\end{pmatrix}\right)$ and $\dist(x, S) \leq 2 \delta$. 
    If $\delta \geq 1/2$, then the property clearly holds by taking an arbitrary $x$ of rank $1$ and using the fact that $0 \in S$. Now assume $\delta < 1/2$.
    The normalisation condition
    can be written as
    \[
    \|a\|_2^2 + \|b\|_2^2 + \|c\|_2^2 + \|d\|^2_2 = 1.
    \]
    Therefore,
    \[1 - (\|b\|_2^2 + \|c\|_2^2) + \dist(b,S)^2 + \dist(c,S^*)^2 = \delta^2.\]
    As the distances between $b$, $c$ and $S$, $S^*$ are non-negative, we end up with:
    \[
    \max\left\{\|b\|_2^2, \|c\|_2^2\right\} \geq \frac{\|b\|_2^2+\|c\|_2^2}{2} \geq \frac{1-\delta^2}{2}.
    \]
    We may assume without loss of generality that $\|b\|_2 \geq \|c\|_2$ (otherwise we exchange the roles of $b$ and $c$). Using additionally $\dist\left(b,S\right) \leq \delta < 1/2$ we get
    \begin{equation*}
    \dist\left(\frac{b}{\|b\|_2}, S\right) \leq \frac{\delta}{\sqrt{\frac{1-\delta^2}{2}}} \leq \sqrt{\frac{8}{3}}\delta \leq 2\delta
    \end{equation*}
    and we have that $\rank(b) \leq \rank\left(\begin{pmatrix} a & b \\ c & d \end{pmatrix}\right)$. This proves the claimed statement and, by taking the infimum over $a,b,c,d$, leads to the second inequality in~\eqref{eq:minrank-s-hats}.

    We now move on to the explicit construction of an orthonormal basis of $\widehat{S}^{\perp}$. Observe that 
    \begin{equation}\label{eq:perp_S_hat}
     \widehat{S}^\perp = E_{00} \otimes \B(\mathbb{C}^n) + E_{01} \otimes S^\perp + E_{10}\otimes S^{\perp *} + E_{11} \otimes \B(\mathbb{C}^m).
     \end{equation}  
    In fact, let $x \in \B(\mathbb{C}^{n+m})$ be written as $x = E_{00} \otimes x_{00} + E_{01} \otimes x_{01} + E_{10} \otimes x_{10} + E_{11} \otimes x_{11}$,

\begin{align*}
x \in \widehat{S}^\perp
&\iff \forall y \in \widehat{S}, \langle x, y \rangle = 0 \\
&\iff \forall y,z \in S : \\
&\quad\ \Bigl\langle E_{00} \otimes x_{00} + E_{01} \otimes x_{01} + E_{10} \otimes x_{10} + E_{11} \otimes x_{11} , E_{01} \otimes y + E_{10} \otimes z^* \Bigr\rangle = 0\\
&\iff \forall y,z \in S, \ \langle x_{00},0 \rangle + \langle x_{11}, 0 \rangle + \langle x_{01}, y \rangle + \langle x_{10}, z^* \rangle = 0 \\
&\iff x \in E_{00} \otimes \B(\mathbb{C}^n) + E_{01} \otimes S^\perp + E_{10}\otimes S^{\perp *} + E_{11} \otimes \B(\mathbb{C}^m)
\end{align*}
where, in the last equivalence, we used the fact that we could choose $y = 0$ (resp. $z = 0$) to conclude that, necessarily, for all $z \in S$, $\langle x_{10}, z^*\rangle = 0$ (resp. for all $y \in S$, $\langle x_{01},y \rangle = 0$).
Now to construct a basis of this subspace, note first that for all $n \in \mathbb{N}$, we can construct an orthonormal Hermitian basis for $\B(\mathbb{C}^n)$. It suffices to take 
\begin{equation*}
\left\{\frac{1}{\sqrt{2}}(E_{p,q} + E_{q,p}), \frac{1}{\sqrt{2}}i(E_{p,q} - E_{q,p}) \;\middle|\; 1 \leq p < q \leq n\right\} \cup \left\{E_{p,p} \; \middle| \; p \in [n]\right\}.
\end{equation*}

Then, we can easily check that $B_{\widehat{S}^\perp}$ is a Hermitian orthonormal basis of $\widehat{S}^\perp$, i.e., that $\operatorname{Span}(B_{\widehat{S}^\perp}) = \widehat{S}^\perp$, all $x \in B_{\widehat{S}^\perp}$ are Hermitian and for all $x,y \in B_{\widehat{S}^\perp}$, $\langle x,y \rangle = \delta_{x = y}$.
\end{proof}

Now $\widehat{S}^{\perp}$ is an operator system and so by \cref{lem:subspace-to-channel}, there exists a quantum channel $\Phi$ such that $G_{\Phi} = \widehat{S}^{\perp}$. By \cref{thm:contraction_coeff_equal_1}, the property $\eta_{\tr}(\Phi) = 1$ exactly captures the property that $\widehat{S}$ contains a rank-one element. To establish \cref{thm:NP_complete_case}, it only remains to quantitatively control some of the steps in the reduction.

    \begin{proof} [Proof of \cref{thm:NP_complete_case}]

    We start with an instance $G$ of Graph 2-CSP. As in \cref{cor:hardness_sep}, we construct $\Pi \in \B(\C^n \otimes \C^n)$ and then a subspace $S$ of $\B(\C^n)$ as in \cref{lem:from-pi-to-s}. Then, we construct $\widehat{S}$ a subspace of $\B(\C^{2n})$ as in \cref{lemma:space_extension} and then let $\Phi : \B(\C^{2n}) \to \B(\C^{\mathrm{poly}(n)})$ be such that $G_{\Phi} = \widehat{S}^{\perp}$ as per \cref{lem:subspace-to-channel}. Note that each step of the construction can be done in polynomial time. A more thorough analysis of the complexity of the reduction steps is provided in Appendix~\ref{sec_apdx:details}. 
    
    For $G$ satisfiable, we have $\hsep(\Pi) = 0$, so $S$ contains a rank-one matrix, so $\widehat{S} = G_{\Phi}^{\perp}$ contains a rank-one matrix and so $\eta_{\tr}(\Phi) = 1$.

    For $G$ $(1-\gamma)$-unsatisfiable, we have $\hsep(\Pi) = \Omega(\frac{1}{n})$ (\cref{cor:hardness_sep}). Then by \cref{lem:from-pi-to-s}, we have that $\inf_{x, \rank(x) = 1, \|x\|_2 = 1} \dist(x,S)^2 = \Omega(\frac{1}{n})$. Using \cref{lemma:space_extension}, we get $\inf_{x, \rank(x) = 1, \|x\|_2 = 1} \dist(x,G_{\Phi}^{\perp})^2 = \Omega(\frac{1}{n})$. By \cref{thm:contraction_coeff_equal_1}, this implies $\eta_{\tr}(\Phi) < 1$, but we want a quantitative bound.

    Let $\Phi$ be the channel constructed from $\widehat{S}^\perp$ by \cref{lemma:quantum_graph}, using the Hermitian orthonormal basis of $\widehat{S}^\perp$ given by \cref{lemma:space_extension}, so that $G_\Phi = \widehat{S}^\perp$, with Kraus operators $\{K_i\}_{i \in [2d']}$ where $d' = \dim(\widehat{S}^\perp) \leq \dim(\B(\C^{2n})) = 4n^2$. Let $u,v \in \C^{2n}$ be unit vectors such that 
    \[
    \eta_{\tr}(\Phi) = \frac{1}{2}\|\Phi(uu^*) - \Phi(vv^*)\|_1 \quad \text{and} \quad v^*u = 0,
    \]
    which exist as per \cref{lemma:Ruskai} (the maximising states can be taken pure and orthogonal). Since $uv^*$ has rank one and $\|uv^*\|_2 = 1$, we have $\dist(uv^*, \widehat{S})^2 = \Omega(\frac{1}{n})$. Moreover, $\tr(uv^*) = v^*u = 0$ and $\|\Pi_{\widehat{S}^\perp}(uv^*)\|_2 = \dist(uv^*,\widehat{S})$ by the last statement of \cref{lemma:quantum_graph}.
    Hence
    \begin{equation}\label{eq:control_Hermitian}
        \sum_{i,j}|\langle K_i^*K_j, uv^*\rangle|^2 \geq \frac{1}{2d'^2}\dist(uv^*,\widehat{S})^2 \geq \frac{1}{32 n^4}\,\Omega\left(\frac{1}{n}\right) = \Omega\left(\frac{1}{n^5}\right).
    \end{equation}

We now use the Fuchs--van de Graaf inequality \cite{Fuchs.1999} to upper-bound the trace distance between $\Phi(uu^*)$ and $\Phi(vv^*)$. For two states $\rho, \sigma$, we write $F(\rho,\sigma) = \|\sqrt{\rho}\sqrt{\sigma}\|_1^2$ for the fidelity between $\rho$ and $\sigma$; the Fuchs--van de Graaf inequality then states that:
\begin{equation}
\frac{1}{2}\|\rho - \sigma\|_1 \leq \sqrt{1 - F(\rho,\sigma)}.
\end{equation}
Therefore, we have:

\begin{equation}
\eta_{\tr}(\Phi) \leq \sqrt{1- F(\Phi(uu^*),\Phi(vv^*))}.
\end{equation}
Then, 
\begin{align*}
F(\Phi(uu^*), \Phi(vv^*)) &= \|\sqrt{\Phi(uu^*)}\sqrt{\Phi(vv^*)}\|_1^2 \\
&\geq \|\sqrt{\Phi(uu^*)}\sqrt{\Phi(vv^*)}\|_2^2 \\
&= \Big\langle \sqrt{\Phi(uu^*)}\sqrt{\Phi(vv^*)}, \sqrt{\Phi(uu^*)}\sqrt{\Phi(vv^*)} \Big\rangle \\
&= \left\langle \Phi(uu^*), \Phi(vv^*)\right\rangle \\
&= \Big\langle \sum_i K_i uu^* K_i^*, \sum_j K_j vv^* K_j^* \Big\rangle \\
&= \sum_{i,j} |\langle  K_i^*K_j, uv^* \rangle|^2 \\
&\geq \Omega\left(\frac{1}{n^5}\right),
\end{align*}
where the last inequality follows from \eqref{eq:control_Hermitian}. Finally,
\begin{equation}
\eta_{\tr}(\Phi) \leq \sqrt{1 - \Omega\left(\frac{1}{n^5}\right)} \leq  1 - \Omega\left(\frac{1}{n^5}\right).
\end{equation}
\end{proof}

\subsection{Quantum Doeblin Coefficients}
\label{sec:quantum_doeblin}

The trace norm (or total-variation) contraction coefficient of a classical channel is upper bounded by one minus the Doeblin coefficient~\cite{doeblin1937proprietes}. In \cite[Theorem 8.17]{Wolf}, a quantum version of Doeblin's theorem was proposed. Quantum Doeblin coefficients were defined and studied in a systematic way more recently in \cite{Hirche2024.2, Hirche2024.1}. 

Recall that for any channel $\Phi$, we denote by $A$ its input system, by $B$ its output system, by $d_A$ and $d_B$ their respective dimensions, and by $J(\Phi)$ the Choi state of $\Phi$.
\begin{definition}[Corollary 3.7 in \cite{Hirche2024.1}]\label{def:QDC}
    Given a channel $\Phi$, the quantum Doeblin coefficient of $\Phi$ is defined as
    \begin{equation}
    \alpha(\Phi) := \sup_{X_B \in \textup{Herm}(B)} \left\{\tr[X_B] \;\middle|\; I_A \otimes X_B \preceq d_A J(\Phi) \right\}.
    \end{equation}
\end{definition}
Furthermore, a relaxation of the quantum Doeblin coefficient, called the \emph{induced Doeblin coefficient}, was also proposed in \cite{Hirche2024.2}.

\begin{definition}[Induced Doeblin coefficient, Proposition 7 in \cite{Hirche2024.2}]\label{def:induced_doeblin}
    Let $\alpha_I(\Phi)$, the \emph{induced Doeblin coefficient}, be defined as 
    \begin{equation}
    \alpha_I(\Phi) = \max_{\substack{ d_AJ(\Phi) -  I_A \otimes X_B \in \textup{Sep}^*(A:B),\\ X_B \in \textup{Herm}(B)}} \tr(X_B),
    \end{equation}
    where $\textup{Sep}^*(A:B)$ is the dual of the cone of separable operators on the compound system $(A, B)$, which corresponds to the cone of block-positive operators. 
\end{definition}

The induced Doeblin coefficient is \emph{not} easy to compute but gives an upper bound on $\alpha(\Phi)$ by Proposition 5 of \cite{Hirche2024.2}.  
Both the Doeblin coefficient and the induced Doeblin coefficient give upper bounds on the trace norm contraction coefficient.

\begin{proposition}[Lemma 3.6 in \cite{Hirche2024.1} and (6.34) in \cite{Hirche2024.2}]
    Let $\Phi$ be a channel, then
        \begin{equation}
        \eta_{\tr}(\Phi) \leq 1 - \alpha_{I}(\Phi) \leq 1 - \alpha(\Phi).
        \end{equation}
\end{proposition}

Therefore, for any channel $\Phi$, whenever $\alpha_{I}(\Phi) \geq \alpha(\Phi) > 0$, we have $\eta_{\tr}(\Phi) < 1$. It is natural to ask whether the converse is true, i.e., does $\eta_{\tr}(\Phi) < 1$ imply $\alpha_{I}(\Phi) > 0$ or $\alpha(\Phi) > 0$? In \cite[Remark 5]{Hirche2024.2}, this question is answered negatively for $\alpha(\Phi)$. We give another entanglement-breaking example of this fact that also answers the question for $\alpha_{I}(\Phi)$.

\begin{proposition}\label{prop:counter_example_Hirche}
    Let $\cH$ be a Hilbert space of dimension $d > 2$ and $\{\ket{i} \mid i \in [d]\}$ an orthonormal basis of $\cH$. The channel $\Phi : \B(\cH) \rightarrow \B(\cH)$ with Kraus operators $\left\{\frac{1}{\sqrt{d-1}} \ketbra{i}{j} \; \middle | \; i \not= j, 1 \leq i,j \leq d\right\}$ is such that:
    \begin{equation}
    \eta_{\tr}(\Phi) < 1 \text{ but } \alpha(\Phi) = \alpha_{I}(\Phi) = 0.
    \end{equation}
\end{proposition}
\begin{proof}
We first show that $\eta_{\tr}(\Phi) < 1$. We have 
\[
G_{\Phi} = \operatorname{Span}\{ \ketbra{j}{i} \ketbra{k}{l}  \mid (i,j,k,l) \in [d]^4, i \not= j, k \not= l\}.
\]
As $d \geq 3$, for all $(j,l) \in [d]^2$, there is an integer $i \in [d]$ such that $i \not= j$, $i \not= l$. Thus, for every $(j,l) \in [d]^2$, $\ketbra{j}{l} \in G_{\Phi}$.
Therefore, $G_\Phi = \B(\cH)$ and thus $G_\Phi^\perp = \{0\}$. Hence, there is no rank-one element in $G_\Phi^\perp$ and by \cref{thm:contraction_coeff_equal_1}, $\eta_{\tr}(\Phi) < 1$.

On the other hand, $J(\Phi) = \frac{1}{d(d-1)} \sum_{i \not= j} \proj{i} \otimes \proj{j}$ and thus
\begin{align}
\label{eq:jphi-idx}
d J(\Phi)- I_{A} \otimes X_B &= \sum_{i} \proj{i} \otimes \left(\frac{1}{d-1}\sum_{j \neq i} \proj{j} - X_B\right).
\end{align}
Positivity of this operator means that for any $i$, $0 \preceq \left(\frac{1}{d-1}\sum_{j \neq i} \proj{j} - X_B\right)$, which implies that $\bra{i} X_B \ket{i} \leq 0$ and so $\tr(X_B) \leq 0$. This means $\alpha(\Phi)=0$. 
In fact, the operator in~\eqref{eq:jphi-idx} is block-diagonal and hence block-positivity implies positivity and we also have $\alpha_{I}(\Phi) = 0$.
\end{proof}

\begin{remark}\label{rem:reviewer}
    Note that the channel given in Proposition~\ref{prop:counter_example_Hirche} is classical, and that it was shown in~\cite[Proposition~9]{Hirche2024.2} that for classical channels $\W$, $\alpha_{\mathrm{cl}}(\W) = \alpha(\W) = \alpha_I(\W)$, with $\alpha_{\mathrm{cl}}(\W) = \sum_y \min_x \W(y | x)$ the classical Doeblin coefficient (see~\cite[(1.3)]{Hirche2024.2}). Hence, the fact that the induced Doeblin coefficient gives a loose bound on the trace norm contraction coefficient (or the total variation contraction coefficient in the classical setting) follows from the fact that the Doeblin coefficient gives a loose bound on this contraction coefficient (see for example~\cite[Remark~5]{Hirche2024.2}).

    In classical notation, we can recast the proof of Proposition~\ref{prop:counter_example_Hirche} as follows.
    Let $\W_{[d] \to [d]}$ be the channel that on input $x$ prepares the uniform distribution on $[d]\backslash\{x\}$, then
    \begin{equation}
        \alpha_I(\W) = \alpha(\W) = \alpha_{\mathrm{cl}}(\W) = \sum_{x \in [d]} \min_{x' \in [d]} \W(x|x') = 0,
    \end{equation}
    \begin{equation}
        \eta(\W) = \max_{x,x'} \frac{1}{2}\| \W(\cdot | x) - \W(\cdot | x')\|_1 = \max_{x,x'} \frac{1}{2}\left\|\frac{1}{d-1} \delta_{x'} - \frac{1}{d-1} \delta_x \right\|_1 = \frac{1}{d-1},
    \end{equation}
    with $\eta(\W)$ the contraction coefficient for the total variation distance, and $\delta_x$ the Dirac distribution at $x$. Hence, for $d \geq 3$, the upper bound $\eta(\W) \leq 1- \alpha(\W)$ is loose.
\end{remark}

\section{Converging hierarchy of efficiently computable upper bounds}\label{sec:sdp_hierarchy}

In this section, we use \cite{Berta.2021} to propose a hierarchy of semidefinite programming upper bounds on $\psucc(\Phi, k)$, which in the special case $k=2$ gives bounds on $\eta_{\tr}(\Phi)$.
Recall that for every channel $\Phi$ and input operator $X_A$, we can express $\Phi(X_A)$ as:
\begin{equation}\label{eq:Choi_to_channel}
\Phi(X_A) = d_A\tr_A(J(\Phi)(X_A^T \otimes I_B)).
\end{equation}

We will introduce copies of the system $B$ that we will denote $B_1, \dots, B_m$. As a shorthand, we write $B_1^m = B_1, \dots, B_m$.
An operator on $AB_1^m$ will be denoted $W_{AB_1^m}$ and we write $W_{B_1^m} = \tr_{A}(W_{AB_1^m})$ and $W_{AB_1^l} = \tr_{B_{l+1}^m}(W_{AB_1^m})$ with $B_{l+1}^m = B_{l+1},\dots,B_m$.  
Furthermore, let $\mathfrak{S}_m$ be the symmetric group on $m$ elements. For $\pi \in \mathfrak{S}_m$, let $\mathcal{U}^\pi_{B_1^m}$ be the unitary which permutes the systems $B_1,\dots,B_m$ according to the permutation $\pi$. Its action on product operators over $B_1^m$ is as follows:
\begin{equation}
\mathcal{U}^\pi_{B_1^m}(W_{B_1} \otimes \dots \otimes W_{B_m}) := W_{B_{\pi(1)}} \otimes \dots \otimes W_{B_{\pi(m)}}.
\end{equation}
A multipartite operator $W_{AB_1^m}$ on $AB_1^m$ is called \emph{symmetric with respect to $A$} if 
\begin{equation}
(\text{Id}_A \otimes \mathcal{U}_{B^m_1}^\pi)(W_{AB_1^m}) = W_{AB_1^m},~\forall \pi \in \mathfrak{S}_m.
\end{equation}

We use the shorthand $j_1^m$ for $j_1, \dots, j_m$, and $j_{\pi(1)}^{\pi(m)}$ for $j_{\pi(1)}, \dots, j_{\pi(m)}$, where $\pi \in \mathfrak{S}_m$.

The following \cref{thm:sdp_POVM} gives a hierarchy of semidefinite programs which forms a non-increasing sequence of upper bounds on $\psucc(\Phi,k)$. We will show in \cref{thm:convergence_sdp} that this hierarchy is a strengthening of the hierarchy proposed in~\cite{Berta.2021} for general constrained bilinear optimisation and thus we can use their convergence result.

\begin{theorem}\label{thm:sdp_POVM}
    For $m,k \in \mathbb{N}^*$, the semidefinite programs 
\begin{align}
&\mathrm{SDP}_m(\Phi,k)= \max \begin{cases} \tfrac{d_A}{k}\sum_{i=1}^k \tr(J(\Phi)W_{AB_1}^{(i,i)}) \textup{ if } m=1, \\[1.5ex]
\tfrac{d_A}{kd_B^{m-1}} \sum_{i\in[k],j_2^m \in [k]^{m-1}}\tr\big(J(\Phi) W_{AB_1}^{(i,ij_2^m)}\big) \textup{ else,}\\
\end{cases}\\
\textup{s.t.} 
 &\; \forall (i,j_1^m) \in [k]^{m+1},~ W^{(i,j_1^m)}_{AB_1^m}\succeq0,\label{eq:cons_positivity} \\
 & \tr\left(\sum_{(i,j_1^m) \in [k]^{m+1}}W_{AB_1^m}^{(i,j_1^m)}\right) = kd_B^m,\label{eq:cons_normalization}\\
 &\; \forall\pi\in\mathfrak{S}_m, W^{\left(i,j_{\pi(1)}^{\pi(m)}\right)}_{AB_1^m}=\left(\mathrm{Id}_A\otimes\mathcal{U}_{B_1^m}^\pi\right)\left(W^{(i,j_1^m)}_{AB_1^m}\right),\label{eq:cons_symetry_invariance}\\
 &\forall (i,j_1^{m-1}) \in [k]^{m},~\sum_{j=1}^k W_{AB_1^m}^{(i,j_1^{m-1}j)} = \sum_{j=1}^kW_{AB_1^{m-1}}^{(i,j_1^{m-1}j)} \otimes \frac{I_B}{d_B},\label{eq:cons_last_system} \\
 &\; \forall (i,j_1^m) \in [k]^{m+1},~W_{B_1^m}^{(i,j_1^m)} = \frac{1}{k}\sum_{l=1}^k W_{B_1^m}^{(l,j_1^m)},\label{eq:cons_first_system}
\end{align}
form a non-increasing sequence of upper bounds on $\psucc(\Phi,k)$. We write $\mathrm{SDP}_m(\Phi) = 2 \mathrm{SDP}_m(\Phi,2) - 1$ for the corresponding upper bound on the contraction coefficient. 
\end{theorem}

Before proving this theorem, we express the success probability $\psucc(\Phi,k)$ defined in \eqref{eq:proba_success} using the Choi state of $\Phi$.

\begin{lemma}\label{lemma:proba_succ_Choi}
We have:
\begin{equation}\label{eq:proba_succ_Choi}
\psucc(\Phi,k) = \max_{\substack{\{M_i \mid i \in [k]\} \mathrm{POVM},\\
\{\rho_i \in\D(\cH) \mid i \in [k]\}}}\frac{d_A}{k}\sum_{i=1}^k \tr(J(\Phi) \rho_i^T \otimes M_i).
\end{equation}
\end{lemma}

\begin{proof}
    We have the following equalities:

\begin{align*}
\psucc(\Phi,k) 
&= \max_{\substack{\{M_i \mid i \in [k]\} \text{ POVM},\\
\{\rho_i \in\D(\cH) \mid i \in [k]\}}} \frac{1}{k}\sum_{i=1}^k \tr(M_i \Phi(\rho_i)) \\
&\overset{(a)}{=} \max_{\substack{\{M_i\mid i \in [k]\} \text{ POVM},\\
\{\rho_i \in\D(\cH)\mid i \in [k]\}}}\frac{d_A}{k}\sum_{i=1}^k \tr\left(M_i\tr_A\left(J(\Phi)\left(\rho_i^T \otimes I_B \right)\right)\right) \\
&= \max_{\substack{\{M_i\mid i \in [k]\} \text{ POVM},\\
\{\rho_i \in\D(\cH)\mid i \in [k]\}}} \frac{d_A}{k}\sum_{i=1}^k \tr\left(J(\Phi) \rho_i^T \otimes M_i\right),
\end{align*}
where $(a)$ follows from \eqref{eq:Choi_to_channel}.
\end{proof}

\begin{proof}[Proof of \cref{thm:sdp_POVM}]
We first show that the semidefinite programs defined in the theorem are relaxations of the optimisation problem in the right-hand side of \eqref{eq:proba_succ_Choi}. 

It suffices to take, at the level $m \in \mathbb{N}^*$ of the hierarchy, the variables $W_{AB_1^m}^{(i,j_1^m)}$ to be equal to the product operator $\rho_i^T \otimes \bigotimes_{p=1}^m M_{j_p}$, with $\{\rho_i\}_i$, $\{M_j\}_j$ being the states and POVM on which the optimisation in \eqref{eq:proba_succ_Choi} is done. Then, the positivity and symmetry constraints \eqref{eq:cons_positivity} and \eqref{eq:cons_symetry_invariance} are trivially satisfied. Furthermore, we have:
\begin{equation*}
\tr\left(\sum_{(i,j_1^m) \in [k]^{m+1}} \rho_i^T\bigotimes_{p=1}^m M_{j_p}\right) = \left(\sum_{i=1}^k \left(\sum_{j=1}^k \tr(M_j)\right)^{m}\right) = kd_B^m,
\end{equation*}
so \eqref{eq:cons_normalization} is satisfied. Then, let $(i,j_1^{m-1}) \in [k]^{m}$ be fixed, we have on the one hand:
\begin{equation*}
\sum_{j=1}^k \rho_i^T \bigotimes_{p=1}^{m-1}M_{j_p}\otimes M_j = \rho_i^T \bigotimes_{p=1}^{m-1}M_{j_p} \otimes \left(\sum_{j=1}^k M_j \right) = \rho_i^T \bigotimes_{p=1}^{m-1}M_{j_p} \otimes I_B;
\end{equation*}
and, on the other hand: 
\begin{equation*}
\sum_{j=1}^k \tr_{B_m}\left(\rho_i^T \bigotimes_{p=1}^{m-1}M_{j_p}\otimes M_j\right) = \sum_{j=1}^k \tr(M_j) \rho_i^T \bigotimes_{p=1}^{m-1}M_{j_p} = d_B\rho_i^T \bigotimes_{p=1}^{m-1}M_{j_p},
\end{equation*}
so these variables satisfy \eqref{eq:cons_last_system}. Finally, for $(i,j_1^m) \in [k]^{m+1}$ fixed, we have:
\begin{align}
\tr_A\left(\rho_i^T \bigotimes_{p=1}^m M_{j_p}\right) = \bigotimes_{p=1}^m M_{j_p} = \frac{1}{k} \sum_{l=1}^k \bigotimes_{p=1}^m M_{j_p} = \frac{1}{k} \sum_{l=1}^k \tr_A\left(\rho_l^T \otimes \bigotimes_{p=1}^m M_{j_p}\right),
\end{align}
so that this choice of variables also satisfies \eqref{eq:cons_first_system}. 
Then, we have to show that, for this choice of variables, the objective function gives the same value as the objective function in the right-hand side of \eqref{eq:proba_succ_Choi}. We present the case $m \geq 2$ as the case $m=1$ is trivial. 
\begin{align*}
    &\frac{d_A}{kd_B^{m-1}} \sum_{i=1}^k\sum_{j_2^m \in [k]^{m-1}}\tr\left(J(\Phi) W_{AB_1}^{(i,ij_2^m)}\right) \\
    &= \frac{d_A}{kd_B^{m-1}}\sum_{i=1}^{k}\tr\left(J(\Phi)\sum_{j_2^m \in [k]^{m-1}}\tr_{B_2^m}\left(\rho_i^T\otimes M_i \bigotimes_{p=2}^m M_{j_p}\right)\right) \\
    &= \frac{d_A}{kd_B^{m-1}}\sum_{i=1}^k\tr\left(J(\Phi)\left(\rho_i^T \otimes M_i\right)\right)\tr\Big(\sum_{j_2^m \in [k]^{m-1}}\bigotimes_{p=2}^m M_{j_p}\Big) \\
    &=\frac{d_A}{kd_B^{m-1}}\tr(I_B)^{m-1}\sum_{i=1}^k\tr\left(J(\Phi)\left(\rho_i^T \otimes M_i\right)\right) \\
    &= \frac{d_A}{k}\sum_{i=1}^k\tr\left(J(\Phi)\left(\rho_i^T \otimes M_i\right)\right).
\end{align*}
Finally, the sequence is non-increasing in $m$, as the operators $\frac{1}{d_B}\sum_{j=1}^k \tr_{B_{m+1}}\left(W^{(i,j_1^m j)}_{AB_1^{m+1}}\right)$ obtained from a feasible solution of $\mathrm{SDP}_{m+1}(\Phi,k)$ are feasible for $\mathrm{SDP}_{m}(\Phi,k)$ (for \eqref{eq:cons_last_system}, use \eqref{eq:cons_symetry_invariance} to exchange $B_m$ and $B_{m+1}$) and achieve the same objective value.
\end{proof}
Note that it is clear from the proof that we can add positive partial transpose (PPT) constraints to the SDP. Indeed, for all $(i,j_1, \dots, j_m) \in [k]^{m+1}$, we can additionally ask the variables $W_{AB_1^m}^{(i,j_1^m)}$ to satisfy 

\begin{equation}\label{eq:ppt_constraints}
    T_A\left(W_{AB_1^m}^{(i,j_1^m)}\right) \succeq 0,~T_{B_1^{\alpha}}\left(W_{AB_1^m}^{(i,j_1^m)}\right) \succeq 0,\forall \alpha \in [m-1],
\end{equation}
where $T_S$ denotes the partial transpose on system $S$. Note that products of positive operators satisfy the PPT constraints of \eqref{eq:ppt_constraints}, therefore the SDP we obtain by adding \eqref{eq:ppt_constraints} is still a relaxation of $\psucc(\Phi,k)$ (see \eqref{eq:proba_success}). We denote $\mathrm{SDP}_{m}^{\mathrm{PPT}}$ the value of the corresponding SDP when adding the PPT constraints.

In the remainder of this section, we show that the hierarchy of semidefinite programs defined in \cref{thm:sdp_POVM} actually converges to $\psucc(\Phi,k)$. Note that, although semidefinite programs are efficiently computable, the number of variables of $\mathrm{SDP}_m(\Phi,k)$ grows exponentially in $m$, while we only prove a convergence rate in $\text{poly}(kd)/\sqrt{m}$; therefore, the convergence of this hierarchy to $\psucc(\Phi,k)$ contradicts neither \cref{thm:NP_contraction_coef} nor \cref{thm:NP_complete_case}.
\begin{theorem}[Convergence of the SDP of \cref{thm:sdp_POVM}]\label{thm:convergence_sdp}
For every quantum channel $\Phi$ and all $k,m \in \mathbb{N}^*$, we have 
\begin{equation}
    0 \leq \mathrm{SDP}_m(\Phi,k) - \psucc(\Phi,k) \leq \frac{\mathrm{poly}(kd)}{\sqrt{m}},
\end{equation}
with $d = \max\{d_A, d_B\}$ with $A$ the input system of $\Phi$ and $B$ its output system, so that 
\begin{equation}
\psucc(\Phi,k) = \lim_{m\to \infty} \mathrm{SDP}_m(\Phi,k).
\end{equation}
\end{theorem}

In order to prove the convergence, we express $\psucc(\Phi,k)$ as a constrained bilinear program~\cite{Berta.2021} and show that the SDP hierarchy we propose in~\cref{thm:sdp_POVM} is a strengthening of the general converging one derived in~\cite{Berta.2021}. We write $\Psi^{(k)}_{\bar{A}\bar{B}} = \sum_{i=1}^k \proj{i} \otimes \proj{i}$ for the (unnormalised) maximally correlated operator between systems $\bar{A}$ and $\bar{B}$, with $k = \dim(\bar{A}) = \dim(\bar{B})$.
\begin{lemma}\label{lemma:P_succ_sdp}
We can express $\psucc(\Phi,k)$ as the following constrained bilinear program:
\begin{align}
&\psucc(\Phi,k) = d_Ad_B\max \tr\left(\left(J(\Phi)_{AB} \otimes \Psi^{(k)}_{\bar{A}\bar{B}}\right)\left(W_{A\bar{A}} \otimes W_{B\bar{B}}\right)\right), \\
&\textup{s.t.}
 ~ W_{A\bar{A}} \succeq 0,~ W_{B\bar{B}} \succeq 0, \\
&\qquad \tr(W_{A\bar{A}}) = 1,~ \tr(W_{B\bar{B}}) = 1, \\
&\qquad \tr_A(W_{A\bar{A}}) = \frac{I_{\bar{A}}}{k}, \\
&\qquad \tr_{\bar{B}}(W_{B\bar{B}}) = \frac{I_B}{d_B}.
\end{align}
\end{lemma}

\begin{proof}
The proof proceeds in two steps. First, we show that we can suppose without loss of generality that the variables $W_{A\bar{A}}$ and $W_{B\bar{B}}$ are block-diagonal with respect to the systems $\bar{A}$ and $\bar{B}$. Then, we show that for such block-diagonal variables, the program of the lemma is equivalent to the optimisation problem of \cref{lemma:proba_succ_Choi}.

For the first part of the proof, recall that $k = \dim(\bar{A}) = \dim(\bar{B})$. Then, we can write
\begin{equation}\label{eq:decompo_variables}
W_{A\bar{A}} = \frac{1}{k}\sum_{i,j=1}^k W_{A}^{(i,j)} \otimes \ketbra{i}{j}_{\bar{A}}, \quad W_{B\bar{B}} = \frac{1}{d_B}\sum_{i,j=1}^kW_B^{(i,j)} \otimes \ketbra{i}{j}_{\bar{B}}.
\end{equation}
From these variables we construct the block-diagonal operators $W_{A\bar{A}}^D$, $W_{B\bar{B}}^D$ obtained from $W_{A\bar{A}}$ and $W_{B\bar{B}}$ by keeping only the diagonal blocks in the decomposition of \eqref{eq:decompo_variables}, i.e.,
\begin{equation}\label{eq:decompo_variables_diag}
W_{A\bar{A}}^D = \frac{1}{k}\sum_{i=1}^k W_A^{(i,i)} \otimes \proj{i}_{\bar{A}},  \quad W_{B\bar{B}}^D = \frac{1}{d_B}\sum_{i=1}^k W_B^{(i,i)} \otimes \proj{i}_{\bar{B}}.
\end{equation}
The objective function of the program of \cref{lemma:P_succ_sdp} evaluated on the two variables $W_{A\bar{A}}$, $W_{B\bar{B}}$ can be written as:
\begin{align}
 &d_Ad_B\tr\left(\left(J(\Phi)_{AB} \otimes \Psi^{(k)}_{\bar{A}\bar{B}}\right)\left(W_{A\bar{A}}\otimes W_{B\bar{B}}\right)\right) \\
&=\frac{d_A}{k}\tr\left(\left(J(\Phi)_{AB} \otimes \left[\sum_{i=1}^k \proj{i}_{\bar{A}} \otimes \proj{i}_{\bar{B}}\right]\right)\left(\left[\sum_{i,j=1}^kW_A^{(i,j)}\otimes \ketbra{i}{j}_{\bar{A}}\right]\otimes\left[ \sum_{i,j=1}^kW_B^{(i,j)}\otimes\ketbra{i}{j}_{\bar{B}}\right]\right)\right) \\ 
&= \frac{d_A}{k}\tr\Bigg(\sum_{i,j,j' = 1}^k J(\Phi)_{AB}\left(W_A^{(i,j)}\otimes W_B^{(i,j')}\right)\otimes \ketbra{i}{j}_{\bar{A}} \otimes \ketbra{i}{j'}_{\bar{B}}\Bigg)\\
&= \frac{d_A}{k}\tr\left(\sum_{i=1}^kJ(\Phi)_{AB}\left(W_A^{(i,i)}\otimes W_B^{(i,i)}\right)\right)\label{eq:reduc_lemma_p_succ} \\
&= d_Ad_B \tr\left(\left(J(\Phi)_{AB}\otimes \Psi^{(k)}_{\bar{A}\bar{B}}\right) \left(W_{A\bar{A}}^D \otimes W_{B\bar{B}}^D\right)\right).
\end{align}
Furthermore, if $W_{A\bar{A}}$ and $W_{B\bar{B}}$ satisfy the constraints of the bilinear optimisation program of the lemma, then so do $W_{A\bar{A}}^D$, $W_{B\bar{B}}^D$. If $W_{A\bar{A}}$ and $W_{B\bar{B}}$ are positive semidefinite, then so are all their diagonal blocks $W_{A}^{(i,i)} \otimes \proj{i}_{\bar{A}}$ and $W_{B}^{(i,i)} \otimes \proj{i}_{\bar{B}}$, thus $W_{A\bar{A}}^D$ and $W_{B\bar{B}}^D$ are constrained to be positive semidefinite. The other conditions involve traces or partial traces and are also easily satisfied by $W_{A\bar{A}}^D$ and $W_{B\bar{B}}^D$ as we have the equalities $\tr(W_{A\bar{A}}) = \tr(W_{A\bar{A}}^D)$, $\tr(W_{B\bar{B}}) = \tr(W_{B\bar{B}}^D)$, $\tr_{A}(W_{A\bar{A}}) = \tr_{A}(W_{A\bar{A}}^D)$ and $\tr_{\bar{B}}(W_{B\bar{B}}) = \tr_{\bar{B}}(W_{B\bar{B}}^D)$. 

Therefore, we can suppose without loss of generality that we optimise the bilinear program on the block-diagonal variables $W_{A\bar{A}}^D$, $W_{B\bar{B}}^D$ of the form given in \eqref{eq:decompo_variables_diag}.

Now, by \cref{lemma:proba_succ_Choi}, we have 
\begin{equation*}
\psucc(\Phi,k) = \max_{\substack{ \{M_i\mid i \in [k]\}~\mathrm{POVM},\\
\{\rho_i \in\D(\cH)\mid i \in [k]\}}}\frac{d_A}{k}\sum_{i=1}^k \tr(J(\Phi) \rho_i^T \otimes M_i).
\end{equation*}
For $\{\rho_i \in \D(\cH)\mid i \in [k]\}$ any set of states and $\{M_i \succeq 0\mid i \in [k]\}$ any POVM, we associate the block-diagonal variables:
\begin{equation}
W_{A\bar{A}}^D = \frac{1}{k} \sum_{i=1}^k \rho_{i,A}^T \otimes \proj{i}_{\bar{A}},\quad W_{B\bar{B}}^D = \frac{1}{d_B}\sum_{i=1}^k M_{i,B} \otimes \proj{i}_{\bar{B}}.
\end{equation}
It is easy to check that these variables satisfy the constraints of the program of \cref{lemma:P_succ_sdp} and, furthermore,

\begin{align*}
\MoveEqLeft d_Ad_B \tr\left[\left(J(\Phi)_{AB} \otimes \Psi^{(k)}_{\bar{A}\bar{B}}\right)\left(W_{A\bar{A}}^D\otimes W_{B\bar{B}}^D\right)\right]\\
&\overset{(a)}{=} \frac{d_A}{k}\tr\left[\sum_{i=1}^kJ(\Phi)_{AB}\left(W_A^{(i,i)}\otimes W_B^{(i,i)}\right)\right]\\
&= \frac{d_A}{k} \tr\left[\sum_{i=1}^k J(\Phi)_{AB} \left(\rho_{i,A}^T \otimes M_{i,B}\right) \right],
\end{align*}
where $(a)$ follows from \eqref{eq:reduc_lemma_p_succ}. Therefore, the bilinear optimisation program proposed in this lemma is an upper bound on $\psucc(\Phi,k)$. 

To show the inequality in the other direction, consider two block-diagonal variables $W_{A\bar{A}}^D$, $W_{B\bar{B}}^D$ as in \eqref{eq:decompo_variables_diag} which satisfy the constraints of the optimisation program. By the previous system of equations, we have:
\begin{equation*}
d_Ad_B \tr\left[\left(J(\Phi)_{AB} \otimes \Psi^{(k)}_{\bar{A}\bar{B}}\right)\left(W_{A\bar{A}}^D\otimes W_{B\bar{B}}^D\right)\right] = \frac{d_A}{k}\tr\left[\sum_{i=1}^k J(\Phi)_{AB}\left(W_A^{(i,i)}\otimes W_B^{(i,i)}\right)\right].
\end{equation*}
As $W_{A\bar{A}}^D \succeq 0$ and $W_{B\bar{B}}^D \succeq 0$, we have that each of their diagonal blocks $W_{A}^{(i,i)}$ and $W_{B}^{(i,i)}$ is positive semidefinite. Then, we have
\begin{align*}
\frac{I_{\bar{A}}}{k} = \tr_{A}\left(W_{A\bar{A}}^D\right)
= \frac{1}{k}\sum_{i=1}^k \tr\left(W_A^{(i,i)}\right) \proj{i}_{\bar{A}}.
\end{align*}
Therefore, for each $i \in [k]$, we have $\tr\left(W_A^{(i,i)}\right) = \bra{i}I_{\bar{A}}\ket{i} = 1$, so that all operators $W_A^{(i,i)}$ are states (or, equivalently, transposes of states). Analogously 
\begin{equation*}
\frac{I_B}{d_B} = \tr_{\bar{B}}\left(W_{B\bar{B}}^D\right) 
= \frac{1}{d_B}\sum_{i=1}^k W_{B}^{(i,i)} \tr\left(\proj{i}_{\bar{B}}\right) 
= \frac{1}{d_B}\sum_{i=1}^k W_{B}^{(i,i)}.
\end{equation*}
Hence, 
$\left\{W_{B}^{(i,i)} \; \middle | \; i \in [k]\right\}$ forms a POVM.
\end{proof}

Now, from the bilinear optimisation program of \cref{lemma:P_succ_sdp}, we can construct a hierarchy of semidefinite programs of the form of (5) in \cite{Berta.2021} that will automatically converge to $\psucc(\Phi,k)$ by the results of \cite[Theorem 3.1]{Berta.2021}. This hierarchy can be written as
\begin{align}
\widetilde{\mathrm{SDP}}_m(\Phi,k) &= d_Ad_B\max \tr\left[ \left(J(\Phi)_{AB} \otimes \Psi^{(k)}_{\bar{A}\bar{B}}\right) \left(W_{A\bar{A}(B\bar{B})_1}\right) \right],\\
\textup{s.t.} &\quad W_{A\bar{A}(B\bar{B})_1^m} \succeq 0,~\tr(W_{A\bar{A}(B\bar{B})_1^m}) = 1, \\
&\quad \forall \pi \in \mathfrak{S}_m,~\left(\text{Id}_{A\bar{A}}\otimes \mathcal{U}_{(B\bar{B})_1^m}^{\pi}\right)\left(W_{A\bar{A}(B\bar{B})_1^m}\right) = W_{A\bar{A}(B\bar{B})_1^m},\\
&\quad \tr_{A}\left(W_{A\bar{A}(B\bar{B})_1^m}\right) = \frac{I_{\bar{A}}}{k} \otimes W_{(B\bar{B})_1^m}, \\
&\quad \tr_{\bar{B}_m}\left(W_{(B\bar{B})_1^m}\right) = W_{(B\bar{B})_1^{m-1}} \otimes \frac{I_{B_m}}{d_B}.
\end{align}
It is easy to see that
\begin{align*}
    \mathrm{SDP}_m(\Phi,k) \leq \widetilde{\mathrm{SDP}}_m(\Phi,k).
\end{align*}
In fact, given a feasible solution $W_{AB_1^m}^{(i,j_1^m)}$ for $\mathrm{SDP}_m(\Phi,k)$, we define 
\begin{equation*}
W_{A\bar{A} (B \bar{B})_{1}^m} = \frac{1}{k d_B^m} \sum_{\substack{i\in [k],\\j_1^m \in [k]^m}} \proj{i}_{\bar{A}} \otimes \proj{j_1^m}_{\bar{B}_1^m} \otimes W_{AB_1^m}^{(i,j_1^m)}.
\end{equation*}
We then check that it achieves the same value for the objective function and it satisfies the constraints. As $\psucc(\Phi,k) \leq \mathrm{SDP}_m(\Phi,k)$, \cref{thm:convergence_sdp} then follows.

\subsection{Numerical Illustration}

We now briefly illustrate this hierarchy on some simple channels. 
We tested the first level of the hierarchy of \cref{thm:sdp_POVM} 
on the channel $\Phi$ of \cref{prop:counter_example_Hirche} with $d = 3$, i.e., the one with vanishing quantum Doeblin coefficients, and obtained a bound of

\begin{equation}\label{eq:upper_bound_contraction_coef_counter_example}
    \eta_{\tr}(\Phi) \leq \frac{1}{2} < 1.
\end{equation}
Thus, this hierarchy already detects at the first level that this channel has a contraction coefficient $<1$.

We now consider (multiple copies of) generalised amplitude damping channels $\mathcal{A}_{p,\eta}$ as in \cite{Hirche2024.1}. Generalised amplitude damping channels are defined via their Kraus operators:
\begin{align*}
    A_1 &= \sqrt{p} 
    \begingroup
    \setlength{\arraycolsep}{2pt}
    \begin{pmatrix} 1 & 0 \\ 0 & \sqrt{\eta} \end{pmatrix}
    \endgroup& A_2&=\sqrt{p}
    \begingroup
    \setlength{\arraycolsep}{2pt}
    \begin{pmatrix}
        0 & \sqrt{1-\eta} \\ 0 & 0
    \end{pmatrix}
    \endgroup \nonumber\\
    A_3&=\sqrt{1-p}
    \begingroup
    \setlength{\arraycolsep}{2pt}
    \begin{pmatrix} \sqrt{\eta} & 0 \\ 0 & 1\end{pmatrix}
    \endgroup
    & A_4&=\sqrt{1-p}
    \begingroup
    \setlength{\arraycolsep}{2pt}
    \begin{pmatrix} 0 & 0 \\ \sqrt{1-\eta} & 0\end{pmatrix}
    \endgroup
    ,
\end{align*}
where $p, \eta \in [0,1]$. We obtain the upper bounds given in Fig.~\ref{fig:amplitude_damping}. Note that this improves on the bounds obtained in~\cite[Fig. 3]{Hirche2024.1}. In fact, it is simple to see that
\begin{equation}
\eta_{\tr}(\mathcal{A}_{p,\eta}) \geq \frac{1}{2}\|\mathcal{A}_{p,\eta}(\proj{+}) - \mathcal{A}_{p,\eta}(\proj{-})\|_1 = \sqrt{\eta},
\end{equation}
with $\ket{+} = \frac{1}{\sqrt{2}}(\ket{0} + \ket{1})$, $\ket{-} = \frac{1}{\sqrt{2}}(\ket{0}-\ket{1})$ so our upper bounds for $\eta_{\tr}(\mathcal{A}_{p,\eta})$ are numerically tight. Furthermore, we tested the first two levels of our hierarchy, with the PPT constraints~\eqref{eq:ppt_constraints}, on two copies of the generalised amplitude damping channel and obtained a tighter upper bound with the second level.

We then also tested the first level of the hierarchy in \cref{thm:sdp_POVM} on (multiple copies of) qubit depolarising channels $\D_p$, for $p \in [0,1]$, defined by:
\begin{equation}\label{eq:depolarizing_channel}
\D_p(\cdot) = (1-p) \text{Id}_2(\cdot) + \tr(\cdot)p\frac{I_2}{2}.
\end{equation}
We obtain numerically tight upper bounds for up to 3 copies of $\D_p$ as shown in Fig.~\ref{fig:depolarizing} (using the PPT constraints~\eqref{eq:ppt_constraints} for 2 and 3 copies). Note that by (3.31) in \cite{Hirche2024.1}, $\eta_{\tr}(\D_p) = 1-p = 1-\alpha(\D_p)$.

The simulations were implemented in Julia using the solver SCS~\cite{scs}.

\begin{figure}[htbp]
    \centering
    \includegraphics[width=\linewidth]{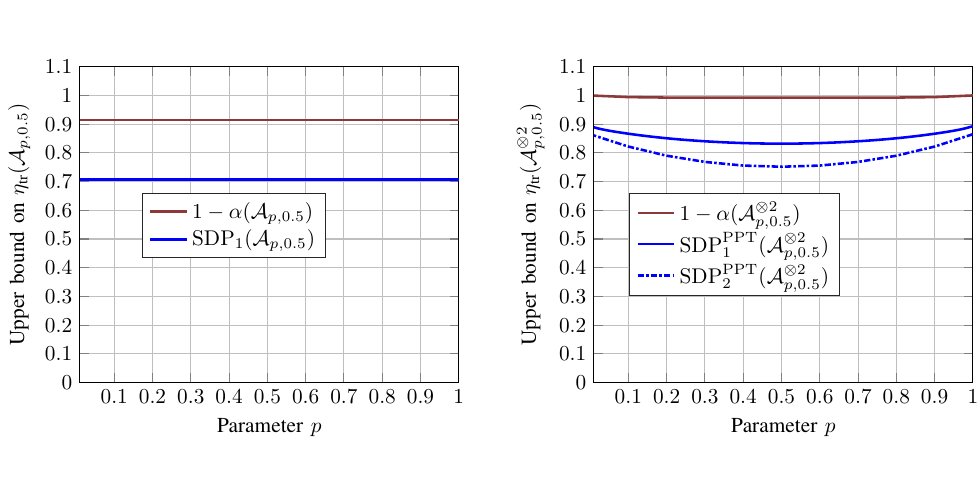} 
    \caption{Upper bounds on the contraction coefficient of (multiple copies of) generalised amplitude damping channels obtained via the first and second levels of the hierarchy proposed in~\cref{thm:sdp_POVM}, with PPT constraints for two copies.}
	\label{fig:amplitude_damping}
\end{figure}

\begin{figure}[htbp]
    \centering
    \includegraphics[width=\linewidth]{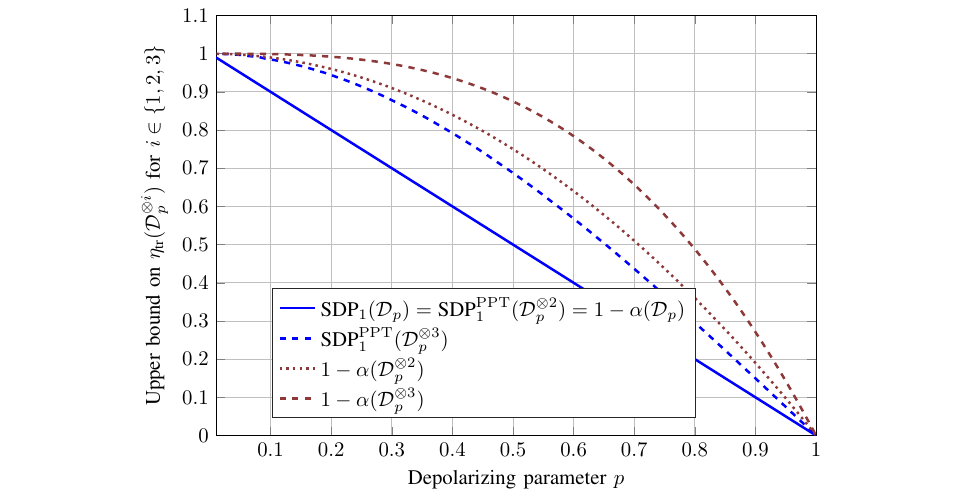}
    \caption{Upper bounds on the contraction coefficient of (multiple copies of) depolarising channels obtained via the first level of the hierarchy proposed in~\cref{thm:sdp_POVM}, with PPT constraints for 2 and 3 copies. By choosing states of the form $\ket{0}^{\otimes i}$ and $\ket{1}^{\otimes i}$, it is simple to see that the bounds given by our SDP are numerically tight.}
     \label{fig:depolarizing}
\end{figure}
\section*{Acknowledgements}
We thank Christoph Hirche for discussions about~\cite{Hirche2024.1} and for helpful comments on the manuscript, Arthur Mehta for discussions about the independence number of quantum confusability graphs and for feedback on the manuscript, and Mizanur Rahaman for discussions on the complexity of the contraction coefficient for some special classes of channels. Moreover, we would like to thank the anonymous reviewers for their helpful comments and constructive suggestions, which improved the quality of this paper. O. Fawzi would also like to thank Mario Berta for numerous discussions over the years about optimal channel coding. This work is funded by the European Research Council (ERC Grant AlgoQIP, Agreement No. 851716).

\bibliographystyle{plainurl}
    \bibliography{bibliography_arxiv.bib}

\appendix

\section{Details on the reduction for Theorem~\ref{thm:NP_complete_case}}\label{sec_apdx:details}

In this section, we show that the reduction given in the proof of Theorem~\ref{thm:NP_complete_case} produces a channel $\Phi$ whose representation, given in Definition~\ref{def:representation_channel}, has a binary length polynomial in that of the instance $G$ of Graph $2$-CSP.

We first show that the acceptance effect $\A_{\cV(G)}$ of the verifier of~\cite[Section~4.1]{LeGall.2012} has rational entries of binary length polynomial in the number of vertices and colours of $G$.

\begin{lemma}\label{lem:acceptance_effect} Given $G$ an instance of Graph $2$-CSP with vertex and colour set $V$ and $\Sigma$ respectively, let $p = |V||\Sigma|$. 
Let $\A_{\cV(G)} \in \C^{p^2 \times p^2}$ be the acceptance effect of the verification procedure given in~\cite[Section 4.1]{LeGall.2012}.
Then, for all $\alpha,\beta \in [p^2]$, $6p^2[\A_{\cV(G)}]_{\alpha,\beta} \in \Z$ and $|6p^2[\A_{\cV(G)}]_{\alpha,\beta}| \leq 6p^2$.
\end{lemma}

\begin{proof}
    Given $u$, $v \in \C^{|V|} \otimes \C^{|\Sigma|}$ two quantum proofs for $G$, the verification protocol in~\cite[Section 4.1]{LeGall.2012} performs one of the following three tests chosen uniformly at random.

    \begin{description}
    \item[\emph{Equality test}]\hspace{1.5cm} Perform the swap test on $u$ and $v$. Reject upon failure.
    \item[\emph{Consistency test}]\hspace{1.5cm} Measure both $u$ and $v$ in the computational basis, which yields pairs of outcomes $(i_u, j_u)$, $(i_v,j_v)$. 
    \begin{itemize}
        \item If $i_u = i_v$, reject if $j_u \neq j_v$.
        \item If $i_u \neq i_v$ and $(i_u,i_v) \in E(G)$, the set of edges of $G$, reject if $R_{(i_u,i_v)}(j_u,j_v) = 0$.
    \end{itemize}
    \item[\emph{Uniformity test}]\hspace{1.5cm} For both $u$ and $v$, apply the quantum Fourier transform on the colour register and measure it in the computational basis. If the outcome is $0$, apply the inverse quantum Fourier transform on the vertex register, measure it in the computational basis, and reject if the outcome is not $0$. Otherwise, in particular if the outcome on the colour register is not $0$, accept.
    \end{description}

    Note that the quantum Fourier transform on $\C^m$ is the unitary $\tfrac{1}{\sqrt{m}}F_m$ with 
    \begin{equation}
    \begingroup
    \setlength{\arraycolsep}{2pt}
        F_m = \begin{pmatrix}
            1 & 1 & 1 & \cdots & 1 \\
        1 & \zeta_m & \zeta_m^2 & \cdots & \zeta_m^{(m-1)} \\
        1 & \zeta_m^2 & \zeta_m^4 & \cdots & \zeta_m^{2(m-1)} \\
        \vdots & \vdots & \vdots & \ddots & \vdots \\
        1 & \zeta_m^{(m-1)} & \zeta_m^{2(m-1)} & \cdots & \zeta_m^{(m-1)(m-1)}
        \end{pmatrix},
    \endgroup
    \end{equation}
    where $\zeta_m = e^{\tfrac{2i\pi}{m}}$ is a primitive $m$-th root of unity (see (4.80) in~\cite{watrous.2018}). Note that $F_m^*F_m = F_mF_m^* = mI$.

    Each of the aforementioned tests amounts to applying unitaries and measurements in the computational basis to $u \otimes v$, and to accepting on some outcomes. 
    In the Heisenberg picture, each test is thus described by an acceptance projector $\Pi_{i,G}$, $i \in \{1,2,3\}$, i.e., the test $i$ accepts $u\otimes v$ with probability $\tr(\Pi_{i,G}\, uu^*\otimes vv^*)$.
    Since the verifier performs only one of the three tests, chosen uniformly at random, the acceptance effect $\A_{\cV(G)}$ is the average of these projectors:
    \begin{equation}
        \A_{\cV(G)} = \frac{1}{3}\left(\Pi_{1,G} + \Pi_{2,G} + \Pi_{3,G}\right).
    \end{equation}

    We explicitly compute $\Pi_{i,G}$ for $i \in \{1,2,3\}$. First, $\Pi_{1,G}$, for the equality test, is the projector onto the symmetric subspace of $(\C^{|V|} \otimes \C^{|\Sigma|})^{\otimes 2}$, i.e.
    \begin{equation}
    \Pi_{1,G} = \frac{1}{2}(I + \mathrm{SWAP}),
    \end{equation}
    where for $\ket{ij} \otimes \ket{i'j'} \in (\C^{|V|} \otimes \C^{|\Sigma|})^{\otimes 2}$, $\mathrm{SWAP} \ket{ij} \otimes \ket{i'j'} = \ket{i'j'} \otimes \ket{ij}$. That is, $\mathrm{SWAP}$ permutes the registers on which $u$ and $v$ are stored.

    Then, the acceptance projector of the consistency test is composed of two terms. One for the case where $i_u$ and $i_v$ are equal, and one for the case where they differ. Since constraint graphs are undirected, we use the convention $R_{(i,i')}(j,j') = R_{(i',i)}(j',j)$ for $(i,i') \in E(G)$.
    Hence, 
    \begin{equation}
        \Pi_{2,G} = \Pi_{\mathrm{eq}} + \Pi_{\mathrm{diff}},
    \end{equation}
    where
    \begin{equation}
        \Pi_{\mathrm{eq}} =  \sum_{i \in V, j \in \Sigma} \proj{ijij},
    \end{equation}
    and
    \begin{equation}
        \Pi_{\mathrm{diff}} =  \sum_{\substack{i_u \neq i_v \in V\\ (i_u,i_v) \notin E(G)}} \sum_{j_u,j_v \in \Sigma} \proj{i_uj_ui_vj_v} + \sum_{\substack{i_u,i_v \in V\\ i_u \neq i_v \\ (i_u,i_v) \in E(G)}} \sum_{\substack{j_u,j_v \in \Sigma \\ R_{(i_u,i_v)}(j_u,j_v) = 1}} \proj{i_uj_ui_vj_v}.
    \end{equation}

    The uniformity test is performed independently on both $u$ and $v$, and accepts if both of them pass it, so that
    \begin{equation}\label{eq:proj_uniformity_test}
        \Pi_{3,G} = \widetilde{\Pi}_{3,G} \otimes \widetilde{\Pi}_{3,G},
    \end{equation}
    where $\widetilde{\Pi}_{3,G}$ is the acceptance projector of the test on a single proof. 
    Measuring a register of dimension $m$ in the computational basis after applying the quantum Fourier transform, resp.\ its inverse, and obtaining the outcome $c$ corresponds to the projector $\tfrac{1}{m}F_m^*\proj{c}F_m$, resp.\ $\tfrac{1}{m}F_m\proj{c}F_m^*$. This test accepts in two disjoint cases. Either the outcome on the colour register, after the Fourier transform, is not $0$, which corresponds to the projector
    \begin{equation}
        \Pi_{I-\proj{0}} = \tfrac{1}{|\Sigma|}\, I \otimes F_{|\Sigma|}^*(I - \proj{0})F_{|\Sigma|}.
    \end{equation}
    Or this outcome is $0$ and the outcome on the vertex register, after the inverse Fourier transform, is $0$, which corresponds to the projector
    \begin{equation}
        \Pi_{\proj{0}} =  \tfrac{1}{|\Sigma||V|}F_{|V|} \proj{0} F_{|V|}^* \otimes F_{|\Sigma|}^*\proj{0}F_{|\Sigma|}.
    \end{equation}
    Hence, $\widetilde{\Pi}_{3,G} = \Pi_{\proj{0}} + \Pi_{I-\proj{0}}$, that is,
    \begin{equation}
        \widetilde{\Pi}_{3,G} = \tfrac{1}{|\Sigma||V|}J_{|V|}\otimes J_{|\Sigma|} + I_{|V|}\otimes\left(I_{|\Sigma|} - \tfrac{1}{|\Sigma|}J_{|\Sigma|}\right),
    \end{equation}
    where $J_m = (1)_{i,j \in [m]}$ is the $m \times m$ matrix whose entries are all $1$, and where we used $F_m\ket{0} = F_m^*\ket{0} = (1,\dots,1)^T$, so that $F_m\proj{0}F_m^* = F_m^*\proj{0}F_m = J_m$, as well as $F_m^*F_m = mI$.

    We now bound the entries of $\A_{\cV(G)}$. The entries of $\Pi_{1,G}$ belong to $\{0,\tfrac{1}{2},1\}$ and those of $\Pi_{2,G}$ to $\{0,1\}$. As $\tfrac{1}{|\Sigma|} = \tfrac{|V|}{p}$, the entries of $\widetilde{\Pi}_{3,G}$ are integer multiples of $\tfrac{1}{p}$, so that those of $\Pi_{3,G}$ are integer multiples of $\tfrac{1}{p^2}$. Hence, $6p^2\A_{\cV(G)} = 2p^2(\Pi_{1,G} + \Pi_{2,G} + \Pi_{3,G})$ has integer entries. Finally, $0 \leq \A_{\cV(G)} \leq I$, so that $\|\A_{\cV(G)}\|_\infty \leq 1$. As the modulus of any entry of a matrix is at most its operator norm, $|[\A_{\cV(G)}]_{\alpha,\beta}| \leq \|\A_{\cV(G)}\|_\infty \leq 1$ for all $\alpha,\beta \in [p^2]$, which gives $|6p^2[\A_{\cV(G)}]_{\alpha,\beta}| \leq 6p^2$.
\end{proof}

We then move on to showing how to compute a basis of $S^\perp$ in Lemma~\ref{lem:from-pi-to-s} from the acceptance effect $\A_{\cV(G)}$.

\begin{lemma}\label{lem:from_A_to_S_perp}
    Let $G$ be an instance of Graph $2$-CSP with vertex set $V$ and colour set $\Sigma$, and $\A_{\cV(G)} \in \B(\C^{|V|^2|\Sigma|^2})$ be the acceptance effect of its verifier in Theorem~\ref{thm:2csp-qma}. 
    For every matrix $B$ row-equivalent\footnote{i.e., $B=Q(I - \A_{\cV(G)})$ for some non-singular matrix $Q$.} to $I -\A_{\cV(G)}$ with linearly independent non-zero rows $(b_j)_{j \in [l]}$, the operators $(\Wvec^{-1}(b_j^*))_{j \in [l]}$ form a basis of $S^{\perp}$, where $S = \Wvec^{-1}(\ker(I - \A_{\cV(G)}))$ is as in Lemma~\ref{lem:from-pi-to-s}.

    Furthermore, if $(b_j)_{j \in [l]}$ are orthonormal, then so are $(\Wvec^{-1}(b_j^*))_{j \in [l]}$.  
\end{lemma}

\begin{proof}
    Let $p = |V||\Sigma|$ and let $(s_j)_{j \in [k]}$ be a basis of $S \subseteq \B(\C^p)$. 
    Then, by definition, $S^\perp = \{x \in \B(\C^p) \mid \langle s_j,x\rangle = 0,~ j \in [k]\}$.
    Since $\Wvec$ preserves the Hilbert--Schmidt product, we can recast $S^{\perp}$ as
    \begin{equation*}
        S^\perp = \left\{\Wvec^{-1}(x),~x \in \C^{p^2} \; \middle | \;  \Wvec(s_j)^*x = 0,~ j \in [k]\right\}.
    \end{equation*}
    
    Hence $(s^{\perp}_j)_{j \in [l]}$ is a basis of $S^{\perp}$ if and only if $(\Wvec(s^\perp_j))_{j \in [l]}$ is a basis of the kernel of 
    \begin{equation}\label{eq:matrix_vec}
       V_S =  \begin{pmatrix}
            \Wvec(s_1)^* \\
            \vdots \\
            \Wvec(s_k)^*
        \end{pmatrix}.
    \end{equation}
    Now, since $S = \Wvec^{-1}(\ker(I - \A_{\cV(G)}))$, $(\Wvec(s_j))_{j \in [k]}$ is a basis of $\ker(I - \A_{\cV(G)})$. 
    If $B$ is row-equivalent to $I - \A_{\cV(G)}$, its non-zero rows $(b_j)_{j \in [l]}$ are such that $(b_j^*)_{j \in [l]}$ form a basis of the image space of $(I - \A_{\cV(G)})^*$, which is the orthogonal complement of $\ker(I - \A_{\cV(G)})$.
    Hence $(b_j^*)_{j \in [l]}$ form a basis of the kernel of $V_S$.
    Therefore $(\Wvec^{-1}(b_j^*))_{j \in [l]}$ form a basis of $S^\perp$. 
\end{proof}

\begin{remark}
    Note that $\Wvec^{-1}(x)$ is obtained solely by reshaping the vector $x \in \C^m \otimes \C^m$ as an $m \times m$ matrix.
     Indeed, if we write $x = \sum_{i,j = 1}^m x_{i,j}\ket{ij}$, then, by linearity, $\Wvec^{-1}(x) = \sum_{i,j=1}^m x_{i,j}\ketbra{i}{j}$.
\end{remark}

Finally, we show the following proposition.

\begin{proposition}\label{prop:satisfaction_computational_model}
    Let $G$ be a Graph $2$-CSP instance, obtained as in~\cite[Theorem~3]{LeGall.2012}, and $\Phi$ the quantum channel constructed from $G$ in the reduction given in the proof of Theorem~\ref{thm:NP_complete_case}.
    Then, the binary length of $\Phi$ under the representation of Definition~\ref{def:representation_channel}, denoted $\bitlength(\Phi)$, is $\bitlength(\Phi) = O(\poly(\bitlength(G))) = O(\poly(|V||\Sigma|))$, where $V$ and $\Sigma$ are the vertex set and the colour set of $G$. 
\end{proposition}

\begin{proof}
We denote $p = |V||\Sigma|$. 
Let $\A_{\cV(G)}$ be the acceptance effect for $G$ as constructed in~\cite[Section~4.1]{LeGall.2012}.
Then, by Lemma~\ref{lem:acceptance_effect}, written in the computational basis, $6p^2\A_{\cV(G)}$ has integer entries, and so does $6p^2(I - \A_{\cV(G)})$, i.e., $6p^2(I - \A_{\cV(G)}) \in \Z^{p^2\times p^2}$. Moreover, since $0 \preceq \A_{\cV(G)} \preceq I$, we also have $0 \preceq I - \A_{\cV(G)} \preceq I$, so that $\|I - \A_{\cV(G)}\|_\infty \leq 1$. As the modulus of any entry of a matrix is at most its operator norm, for all $\alpha,\beta \in [p^2]$, $|6p^{2}[I-\A_{\cV(G)}]_{\alpha,\beta}| \leq 6p^2\|I - \A_{\cV(G)}\|_\infty \leq 6p^2$.
Let $(a_j)_{j \in [l]}$, with $l \leq p^2$, be a maximal linearly independent family of rows of $6p^{2}(I - \A_{\cV(G)})$ and let $(\tilde{a}_j)_{j \in [l]}$ denote the orthogonal vectors obtained by applying the Gram--Schmidt process to $(a_j)_{j \in [l]}$, i.e., for all $j \in [l]$
    \begin{equation}\label{eq:gram_schmidt}
        \tilde{a}_j =  a_j - \sum_{\alpha =1}^{j-1} \frac{\langle \tilde{a}_\alpha,a_j \rangle}{\langle \tilde{a}_\alpha,\tilde{a}_\alpha \rangle}\tilde{a}_\alpha.
    \end{equation}
    Hence, $\tilde{a}_j$, for $j \in [l]$, have rational coefficients.
    Furthermore, by Proposition~1.26 in the seminal article by Lenstra, Lenstra and Lov{\'a}sz~\cite{LLL}, each of these coefficients has binary length $O(p^2\log(p))$. 
    Then, by Lemma~\ref{lem:from_A_to_S_perp} and Lemma~\ref{lemma:space_extension}, if we denote $\widehat{S}^\perp$ the operator system defined in Lemma~\ref{lemma:space_extension}, the set of matrices $\cup_{\alpha=0}^6 B_\alpha $ forms an orthonormal basis of $\widehat{S}^\perp$, with
    \begin{align}
        \nonumber B_0 &=\left\{E_{00}\otimes\ketbra{\alpha}{\alpha},\ E_{11}\otimes\ketbra{\alpha}{\alpha} \;\middle|\; \alpha \in [p] \right\}, \\
        B_1 &=\left\{\frac{1}{\sqrt{2}}
        \begingroup
        \setlength{\arraycolsep}{2pt}
        \begin{pmatrix} 1 & 0 \\ 0 & 0 \end{pmatrix} 
        \endgroup
       \otimes (\ketbra{\alpha}{\beta}+\ketbra{\beta}{\alpha})  \middle|  1 \leq \alpha<\beta \leq p \right\}, \\
        B_2 &= \left\{\frac{i}{\sqrt{2}}
        \begingroup
        \setlength{\arraycolsep}{2pt}
        \begin{pmatrix} 1 & 0 \\ 0 & 0 \end{pmatrix} 
        \endgroup
        \otimes (\ketbra{\alpha}{\beta} -\ketbra{\beta}{\alpha}) \middle|1 \leq \alpha<\beta \leq p \right\}, \\
        B_3 &= \left\{\frac{1}{\sqrt{2}}
        \begingroup
        \setlength{\arraycolsep}{2pt}
        \begin{pmatrix} 0 & 0 \\
        0 & 1 \end{pmatrix} 
        \endgroup
        \otimes (\ketbra{\alpha}{\beta} + \ketbra{\beta}{\alpha}) 
        \middle| 1 \leq \alpha<\beta \leq p \right\},  \\
        B_4 &= \left\{\frac{i}{\sqrt{2}}
        \begingroup
        \setlength{\arraycolsep}{2pt}
        \begin{pmatrix} 0 & 0 \\
        0 & 1 \end{pmatrix}
        \endgroup
        \otimes (\ketbra{\alpha}{\beta} -\ketbra{\beta}{\alpha}) \middle|1 \leq \alpha<\beta \leq p \right\},\\
        B_5 &= \left\{\frac{1}{\sqrt{2\langle \tilde{a}_j, \tilde{a}_j\rangle}}
        \begingroup
        \setlength{\arraycolsep}{-4pt}
        \begin{pmatrix} 0 & \Wvec^{-1}(\tilde{a}_j^*) \\
            \Wvec^{-1}(\tilde{a}_j^*)^* & 0
            \end{pmatrix}
        \endgroup
            \middle| j \in [l]\right\}, \\
        B_6 &= \left\{\frac{i}{\sqrt{2\langle \tilde{a}_j, \tilde{a}_j\rangle}}
        \begingroup
        \setlength{\arraycolsep}{-4pt}
        \begin{pmatrix} 0 &  \Wvec^{-1}(\tilde{a}_j^*) \\ -\Wvec^{-1}(\tilde{a}_j^*)^* & 0  \end{pmatrix}
        \endgroup
        \middle| j \in [l] \right\}.
    \end{align}
    Indeed, since for all $X \in \B(\C^p)$, $X = \tfrac{1}{2}(X + X^*) - \tfrac{i}{2}(iX - iX^*)$, $\{\ketbra{\alpha}{\alpha}\}_{\alpha \in [p]} \cup \left\{\tfrac{1}{\sqrt{2}}(\ketbra{\alpha}{\beta} + \ketbra{\beta}{\alpha})\right\}_{\alpha<\beta} \cup \left\{\tfrac{i}{\sqrt{2}}(\ketbra{\alpha}{\beta} - \ketbra{\beta}{\alpha})\right\}_{\alpha<\beta} $ forms a Hermitian orthonormal basis of $\B(\C^p)$.

    Let $B = \cup_{i=0}^6B_i$ and $M \in B$. 
    Then, $M$ is a $2p \times 2p$ Hermitian matrix with entries either in $\Q\left(i, \tfrac{1}{\sqrt{2}}\right)$ or in $\Q\left(i,\tfrac{1}{\sqrt{2}}, \tfrac{1}{\sqrt{\langle \tilde{a}_j, \tilde{a}_j\rangle}}\right)$.
    Therefore, $M$ can be represented as in Definition~\ref{def:rep_matrix_number_field}.
    In the first case, i.e., $M \in \cup_{\alpha=0}^4 B_\alpha$ and $M \in \Q\left(i, \tfrac{1}{\sqrt{2}}\right)^{2p \times 2p}$, the entries of $M$ are independent of $G$, so that the binary length of the representation of $M$ is $O(p^2)$.
    In the second case, $M$ is represented by a list $(T_1,T_2,T_3,C_{1,1},\dots,C_{2p,2p})$. 
    $T_1$ and $T_2$ represent $i$ and $\tfrac{1}{\sqrt{2}}$ respectively, hence their binary length is $O(1)$.
    $T_3$ represents $\tfrac{1}{\sqrt{\langle \tilde{a}_j, \tilde{a}_j\rangle}} = \sqrt{\tfrac{q_j}{p_j}}$, where $\langle \tilde{a}_j, \tilde{a}_j\rangle = \tfrac{p_j}{q_j}$ with $\gcd(p_j,q_j)=1$.
    Since the minimal polynomial of $\sqrt{q_j/p_j}$ divides $p_j X^2 - q_j$, the binary length of $T_3$ is at most proportional to that of $\langle \tilde{a}_j, \tilde{a}_j\rangle$, which is $O(p^2 \log(p))$ by~\cite{LLL}.
    Then, for $\alpha,\beta \in [2p]$, $C_{\alpha,\beta}$ is either empty, i.e., $[M]_{\alpha,\beta} = 0$, or $C_{\alpha,\beta} = (B_{\alpha,\beta}, D_{\alpha,\beta})$, where $B_{\alpha,\beta} = (1,2,3)$ has binary length $O(1)$, and $D_{\alpha,\beta}$ contains at most $8$ numerators and one denominator, exactly one numerator being non-zero, and these are, up to sign, the numerator and the denominator of a coefficient of $\tilde{a}_j$. 
    It hence has binary length $O\left(p^2 \log(p)\right)$.
    In total, the binary length of $C_{\alpha,\beta}$ is $O\left(p^2 \log(p)\right)$.
    Overall, the binary length of $M$ is then $O\left(p^4 \log(p)\right)$.

    Let $l' = |B|$ and, for $M \in B$, let $P_{\pm} := \frac{1}{2l'}(I \pm M)$, which are the positive operators of \cref{lemma:quantum_graph} (here $d = l'$). The entries of $P_\pm$ are those of $\pm M$ up to an additive rational number on the diagonal and a multiplicative rational factor $\frac{1}{2l'}$, so that $P_\pm$ belongs to the same number field as $M$ and its binary length is $\bitlength\left(P_\pm\right) = O\left(\bitlength(M) + p^2\log p\right) = O\left(p^4\log p\right)$.

    Now, let us write the spectral decomposition of $P \in \{P_+,P_-\}$ as
    \begin{equation}
        P = \sum_{\alpha=1}^{\rank(P)} \lambda_\alpha \ketbra{\psi_\alpha}{\psi_\alpha},
    \end{equation}
    where $(\lambda_\alpha)_{\alpha \in [\rank(P)]}$ are the (positive) non-zero eigenvalues of $P$ counted with multiplicities and $(\ket{\psi_\alpha})_{\alpha \in [\rank(P)]}$ is an orthonormal basis of eigenvectors of $P$.
    Then, by Proposition~\ref{prop:rep_spectral_decompo}, for each $\alpha \in [\rank(P)]$, the matrix $\lambda_{\alpha}\ketbra{\psi_\alpha}{\psi_\alpha}$ admits a representation of binary length $O(\poly(\bitlength(P))) = O(\poly(p))$.
    Indeed, as $P$ is always expressed using at most $3$ generators each of degree at most $2$, namely $i$, $1/\sqrt{2}$ and $1/\sqrt{\langle \tilde{a}_j, \tilde{a}_j\rangle}$, the exponential factor in \eqref{eq:complexity_spectral_decomposition} is independent of $G$ and can be treated as a constant.
    Likewise, by the proof of Proposition~\ref{prop:rep_spectral_decompo} with the additional generator $\sqrt{\lambda_\alpha}$, the matrix $\sqrt{\lambda_\alpha}\ketbra{\alpha}{\psi_\alpha}$ has binary length $O(\poly(p))$ under the representation of Definition~\ref{def:rep_matrix_number_field}.
    Then, the Kraus operator of $\Phi$ associated to $P$ is given by \eqref{eq:kraus_from_basis}, which we recall here:
    \begin{equation*}
        K = \sum_{\alpha=1}^{\rank(P)} \sqrt{\lambda_\alpha}\ketbra{\alpha}{\psi_\alpha}.
    \end{equation*}

    Most importantly, each coefficient of $K$ belongs to a unique term $\sqrt{\lambda_\alpha}\ketbra{\alpha}{\psi_\alpha}$.
    Hence, each coefficient of $K$ is represented by $(B_{\beta,\gamma},D_{\beta,\gamma})$, as in Definition~\ref{def:rep_matrix_number_field}, and the number of elements of $B_{\beta,\gamma}$ does not depend on $G$.
    Therefore, the binary length of the representation of $K$ is $O(\poly(p))$.
    $\Phi$ is then represented by concatenating the representations of the $2l' = O(p^2)$ Kraus operators $K$, as in Definition~\ref{def:representation_channel}.
    Its binary length is thus $O(\poly(p))$.
\end{proof}

\section{Details on Definition~\ref{def:representation_channel}}\label{sec_apdx:comput_model}

In this section, we prove Proposition~\ref{prop:rep_spectral_decompo} and we provide Corollary~\ref{coro:mignotte} which proves that the binary length of the approximant of $\theta$ is polynomial in that of its minimal polynomial.

We begin with the proof of Proposition~\ref{prop:rep_spectral_decompo}. 
\begin{proof}[Proof of Proposition~\ref{prop:rep_spectral_decompo}]
    Let $\theta_1,\dots,\theta_k \in \overline{\Q}$, and $A \in \C^{n \times n}$ be a Hermitian matrix with entries in $\Q(\theta_1,\dots,\theta_k)$. Then, its eigenvalues $\lambda_1,\dots,\lambda_n$ (counted with multiplicities) are the roots of its characteristic polynomial $\chi_A$. Since $\chi_A = \det(XI - A)$, its coefficients are in $\Q\left(\theta_1,\dots,\theta_k\right)$. Hence the $(\lambda_i)_{i \in [n]}$ are algebraic numbers. Then, the eigenvectors of $A$ are obtained by computing a basis of $\ker(A - \lambda_j I)$ in $\Q\left(\theta_1,\dots,\theta_k,\lambda_j\right)$, using Gaussian elimination over this field. Since $A$ is Hermitian and $\lambda_j \in \R$, the field generated by the entries of $A$ and $\lambda_j$ is closed under complex conjugation and contains the reduced row echelon form of $A - \lambda_j I$; hence the inner products below remain in this field.
    This basis is in general not orthogonal, but we can apply the Gram--Schmidt procedure to obtain an orthogonal basis $(v_\alpha^{(j)})_{\alpha \in [m_j]}$ of $\ker(A - \lambda_j I)$ whose coefficients are still in $\Q(\theta_1,\dots,\theta_k,\lambda_j)$. 
    We then normalise this basis $(v_\alpha^{(j)})_{\alpha \in [m_j]}$.
    For all $\alpha \in [m_j]$, the norm of $v^{(j)}_\alpha$ is $\sqrt{\langle v_\alpha^{(j)}, v_\alpha^{(j)}\rangle}$ which has degree at most $2$ over $\Q(\theta_1,\dots,\theta_k,\lambda_j)$. 
    If we denote $\ket{\psi_\alpha^{(j)}} = \tfrac{1}{\sqrt{\langle v_\alpha^{(j)}, v_\alpha^{(j)}\rangle}}v_\alpha^{(j)}$, it is then straightforward to express its coefficients in $\Q\left(\theta_1,\dots,\theta_k,\lambda_j,\sqrt{\langle v_\alpha^{(j)}, v_\alpha^{(j)}\rangle}^{-1}\right)$.

    Although the description of each entry of $A$ may require only a subset of the generators $(\theta_j)_{j \in [k]}$, the computations above involve all of them, as well as $\lambda_j$ and $\sqrt{\langle v_{\alpha}^{(j)}, v_{\alpha}^{(j)}\rangle}^{-1}$.
    Writing $d = \max_{j \in [k]} \deg(\theta_j)$, we have $\deg(\lambda_j) \leq n d^k$, since $\lambda_j$ is a root of $\chi_A$, of degree $n$ with coefficients in $\Q(\theta_1,\dots,\theta_k)$, and similarly $\deg\left(\sqrt{\langle v_{\alpha}^{(j)}, v_{\alpha}^{(j)}\rangle}^{-1}\right) \leq 2 n d^k$.
    Hence, by Definition~\ref{def:rep_matrix_number_field}, each entry of $\lambda_j \ketbra{\psi_{\alpha}^{(j)}}{\psi_\alpha^{(j)}}$ is represented by at most $2n^2 d^{3k} + 1$ integers.
    Moreover, the computations of $\chi_A$, of the minimal polynomial of $\lambda_j$ (as an irreducible factor of the norm of $\chi_A$, obtained by the factorisation algorithm of~\cite{LLL}), of the kernels by Gaussian elimination~\cite{Bareiss.1968}, and the Gram--Schmidt procedure only produce integers of binary length polynomial in $n d^k$ and $\bitlength(A)$.
    Therefore, the representation of $\lambda_j \ketbra{\psi_{\alpha}^{(j)}}{\psi_\alpha^{(j)}}$ has a binary length of $d^{O(k)} \cdot \poly(\bitlength(A))$.
\end{proof}

We end this Appendix by showing that in the absolute representation of $\theta \in \overline{\Q}$, the approximation $\theta_\app$ can be chosen of binary length polynomial in that of $P_\theta$.

To this end, we first recall the following lemma due to Maurice Mignotte.

\begin{lemma}[{Mignotte's bound~\cite[Corollary~3]{Mignotte.1995}}]\label{lem:mignotte_bound}
Let $\alpha, \beta$ be two distinct roots of $P = \sum_{j=0}^d c_jX^j \in \Z[X]$ with $d \geq 2$, then
\begin{equation}\label{eq:mignotte}
   |\alpha - \beta| \geq \sqrt{6}d^{-\tfrac{d}{2}}((d-1)(2d-1))^{-\tfrac{1}{2}}\|P\|_{\ell_2}^{-(d-1)}.
\end{equation}
\end{lemma}

Hence, we have the following corollary.

\begin{corollary}\label{coro:mignotte}
   Let $P \in \Z[X]$ of degree $d$, let $x\in \C$ be a root of $P$, then there exists $x_\app \in \C$ such that 
   \begin{itemize}
       \item For all $y \neq x$ such that $P(y) = 0$, 
       \begin{equation}
           |x - x_{\app}| < \frac{|x-y|}{2}.
       \end{equation}
       \item $x_{\app} = \frac{a}{c} + i\frac{b}{e}$, $a,b,c,e \in \Z$ and
       \begin{equation}\label{eq:bit-length-approximant}
       \operatorname{bit-length}(x_\app) =  O(\log_{2}((|a|+1)(|b|+1)\,c\,e + 1)) = O(d\log(d\|P\|_{\ell_2} +1)).
       \end{equation}
   \end{itemize}
\end{corollary}

\begin{proof}
   If $d=1$, then $x\in\Q$ and one takes $x_\app = x$; we thus assume $d \geq 2$. By density of the rationals in the set of real numbers, we can choose $x_\app = \tfrac{a}{c} + i\tfrac{b}{e}$ arbitrarily close to $x$ but we want to control the bit-length of this approximation. We decompose $x$ into its real and imaginary parts $x = x_{\mathfrak{r}} + ix_{\mathfrak{i}}$ and we write 
   \begin{equation}
       M = \sqrt{6}d^{-\tfrac{d}{2}}((d-1)(2d-1))^{-\tfrac{1}{2}}\|P\|_{\ell_2}^{-(d-1)}.
   \end{equation}
   Then, by the Dirichlet Approximation Theorem (see for reference \cite[Theorem~1A]{Schmidt.1980} and its subsequent Remark), we can choose $c \in \mathbb{N}^*$, with $c < 3\sqrt{2}M^{-1}$ such that there exists $a \in \Z$ satisfying
   \begin{equation}\label{eq:approximation_quality}
       |x_{\mathfrak{r}} - \frac{a}{c}| \leq \frac{1}{3\sqrt{2}M^{-1}c} \leq \frac{1}{3\sqrt{2}}M,
   \end{equation}
   where the second inequality follows from the fact that $c$ is a non-zero integer and hence $c \geq 1$. Then, by the triangle inequality $||x_{\mathfrak{r}}| - |\frac{a}{c}|| \leq |x_{\mathfrak{r}} - \frac{a}{c}|$, 
   \begin{equation}
       |a| \leq c\left(|x_{\mathfrak{r}}| + \frac{1}{3\sqrt{2}}M\right) \leq 3\sqrt{2}M^{-1}\left(|x_{\mathfrak{r}}| + \frac{1}{3\sqrt{2}}M\right) = 3\sqrt{2}M^{-1}|x_{\mathfrak{r}}| + 1.
   \end{equation}
   We then apply Cauchy's bound on the roots of $P$. Indeed, since $P(x)=0$, we have 
    \begin{equation}
       |x_{\mathfrak{r}}| \leq |x| \leq 1 + \max_{0 \leq j \leq d-1}\frac{|c_j|}{|c_d|} \leq 1 + \|P\|_{\ell_2}.
   \end{equation}
   Hence, 
    \begin{equation}\label{eq:upper_bound_numerator_approx_root}
       |a| \leq 3\sqrt{2}M^{-1}(1+\|P\|_{\ell_2}) + 1.
   \end{equation}

   The same line of argument holds, mutatis mutandis, for the imaginary part of $x$, $x_\mathfrak{i}$, and its approximation $\tfrac{b}{e}$. Applying the logarithm to the inequalities $c < 3\sqrt{2}M^{-1}$, $e < 3\sqrt{2}M^{-1}$, and \eqref{eq:upper_bound_numerator_approx_root} (which holds for both $|a|$ and $|b|$) yields \eqref{eq:bit-length-approximant}. Furthermore, for these chosen integers $a,b,c,e$ and $x_{\app} = \tfrac{a}{c} + i \tfrac{b}{e}$, we have by \eqref{eq:approximation_quality} and by Mignotte's bound (Lemma~\ref{lem:mignotte_bound}), 
    \begin{equation}
       |x  - x_{\app}| \leq \frac{1}{3}M < \frac{|x - y|}{2},
    \end{equation}
    where $y \in \C$ is any root of $P$ distinct from $x$.
\end{proof}

\end{document}